\documentclass{siamart171218} % siam
\usepackage{amsfonts}
\usepackage{graphicx}
\usepackage{epstopdf}
\usepackage{algorithmic}
\usepackage{amssymb}
\usepackage{mathtools}
\usepackage{color}
\usepackage[shortlabels]{enumitem}
\usepackage{multirow}
\usepackage{xfrac}
\usepackage[super]{nth}
\usepackage{cite}
\usepackage{icomma}
\usepackage{array}
\usepackage{tabularx}
\usepackage{arydshln}
\usepackage{hyperref}
\usepackage{booktabs}
\usepackage{xspace}

\usepackage{bm}

\newcommand{\mat}[1]{\ensuremath{\boldsymbol{#1}}}
\newcommand{\A}{\ensuremath{\mat{A}}}

\renewcommand{\H}{\ensuremath{\mat{H}}}
\newcommand{\I}{\ensuremath{\mat{I}}}

\newcommand{\M}{\ensuremath{\mat{M}}}

\newcommand{\T}{\ensuremath{\mat{T}}}
\newcommand{\U}{\ensuremath{\mat{U}}}

\renewcommand{\arraystretch}{1.1}

\allowdisplaybreaks

\ifpdf
  \DeclareGraphicsExtensions{.eps,.pdf,.png,.jpg}
\else
  \DeclareGraphicsExtensions{.eps}
\fi

\makeatletter
\newsavebox\myboxA
\newsavebox\myboxB
\newlength\mylenA
\newcommand*\xoverline[2][0.75]{%
    \sbox{\myboxA}{$\m@th#2$}%
    \setbox\myboxB\null% Phantom box
    \ht\myboxB=\ht\myboxA%
    \dp\myboxB=\dp\myboxA%
    \wd\myboxB=#1\wd\myboxA% Scale phantom
    \sbox\myboxB{$\m@th\overline{\copy\myboxB}$}%  Overlined phantom
    \setlength\mylenA{\the\wd\myboxA}%   calc width diff
    \addtolength\mylenA{-\the\wd\myboxB}%
    \ifdim\wd\myboxB<\wd\myboxA%
       \rlap{\hskip 0.5\mylenA\usebox\myboxB}{\usebox\myboxA}%
    \else
        \hskip -0.5\mylenA\rlap{\usebox\myboxA}{\hskip 0.5\mylenA\usebox\myboxB}%
    \fi}
\makeatother

\newcommand{\HH}[3]{\ensuremath{\mathcal{H}_{\{#1\}}^{#2 \times #2}(#3)}}

\newcommand{\ZH}[3]{\ensuremath{\mathcal{Z}_{\{#1\}}^{#2 \times #2}(#3)}}

\newcommand{\MM}[3]{\ensuremath{\mathcal{M}_{\{#1\}}^{#2 \times #2}(#3)}}

\newcommand{\MMZ}{\ensuremath{\mathcal{M}^{n \times n}}\xspace}
\newcommand{\MMM}{\ensuremath{\xoverline{\mathcal{M}}^{n \times n}}\xspace}

\renewcommand{\proofbox}{$\natural$}

\newsiamremark{remark}{Remark}
\newsiamthm{conjecture}{Conjecture}
\headers{Bohemian Upper Hessenberg and Toeplitz Matrices}{E. Y. S. Chan, et al.}

\title{Upper Hessenberg and Toeplitz Bohemians}

\author{
  Eunice Y. S. Chan\thanks{School of Mathematical and Statistical Sciences, Western University
    (\email{echan295@uwo.ca}, 
    \email{rcorless@uwo.ca},
    \email{sthornt7@uwo.ca}).}
  \and
  Robert M. Corless\footnotemark[1]
  \and
  Laureano Gonzalez-Vega\thanks{Departamento de Matematicas, Estadistica y Computacion, Universidad de Cantabria
  (\email{laureano.gonzalez@unican.es}).}
  \and
  J.~Rafael Sendra\thanks{Research Group ASYNACS, Departamento de F{\'i}sica y Matem\'{a}ticas, University of Alcal{\'a}
  (\email{rafael.sendra@uah.es}).}
  \and
  Juana Sendra\thanks{Universidad Polit{\'{e}}cnica de Madrid
  (\email{jsendra@etsist.upm.es}).}
  \and
  Steven E. Thornton\footnotemark[1]
}

\usepackage{amsopn}

\begin{document}

\maketitle

% ============================================================================ %
% Abstract                                                                     %
% ============================================================================ %
\begin{abstract}
We look at Bohemians, specifically those with population $\{-1, 0, {+1}\}$ and sometimes $\{0,1,i,-1,-i\}$. More, we specialize the matrices to be upper Hessenberg Bohemian. 
From there, focusing on only those matrices whose characteristic polynomials have maximal height allows us to explicitly identify these polynomials and give useful bounds on their height, and conjecture an accurate asymptotic formula. The lower bound for the maximal characteristic height is exponential in the order of the matrix; in contrast, the height of the matrices remains constant.  We give theorems about the numbers of normal matrices and the numbers of stable matrices in these families.
\end{abstract}

\begin{keywords}
    upper Hessenberg, Toeplitz, characteristic polynomial, Bohemians, maximal characteristic height, normal matrices, stable matrices
\end{keywords}

\begin{AMS}
  15B05, 15B36, 11C20
\end{AMS}

% ============================================================================ %
% Introduction                                                                 %
% ============================================================================ %
\section{Introduction}
A matrix family is called \textbf{Bohemian} if its entries come from a fixed finite discrete (and hence bounded) set, called the \textsl{population}, usually of integers. The name is a mnemonic for \textbf{Bo}unded \textbf{He}ight \textbf{M}atrix of \textbf{I}ntegers. Such populations arise in many applications (e.g.~compressed sensing) and the properties of matrices selected ``at random'' from such families are of practical and mathematical interest. For example, Tao and Vu have shown that random matrices (more specifically real symmetric random matrices in which the upper-triangular entries $\xi_{i, j}$, $i < j$ and diagonal entries $\xi_{i, i}$ are independent) have simple spectrum~\cite{tao2017random}. In fact, Bohemian families have been studied for a long time, although not under that name. For instance, Olga Taussky-Todd's paper ``Matrices of Rational Integers"~\cite{taussky1960matrices} begins by saying
\begin{quote}
    ``This subject is very vast and very old. It includes all of the arithmetic theory of quadratic forms, as well as many of other classical subjects, such as latin squares and matrices with elements $+1$ or $-1$ which enter into Euler's, Sylvester's or Hadamard's famous conjectures."
\end{quote}
Taussky-Todd also discusses such problems in~\cite{taussky1961some}.
The paper~\cite{gear1969simple} by C.~W.~Gear is another instance.  The idea is that these families are themselves interesting objects of study, and susceptible to brute-force computational experiments (both ideas already in~\cite{taussky1960matrices}) as well as to asymptotic analysis.  Such experiments have generated many conjectures, some of which are listed on the Characteristic Polynomial Database~\cite{CPDB}. Many of the conjectures have a number-theoretic or combinatorial flavour. The recent paper~\cite{Fasi:2019:determinants} claims to resolve several of these conjectures.

Matrices with population $P = \{-1,0,1\}$ occur naturally as exemplars of ``sign-pattern matrices'': see~\cite{briat2017sign} and~\cite{Hall:HLA:2013}. For early theorems, see~\cite{jeffries1977matrix}.

An overview of some of our original interest in Bohemians can be found in~\cite{corless2017bohemian}.  More details and reasons why such matrices are interesting are discussed in section~\ref{sec:Motivation}.

Different matrix structures produce remarkably different results. One matrix structure useful in eigenvalue computation is the upper Hessenberg matrix, which means a matrix $\H$ such that $h_{i, j} = 0$ if $i > j + 1$. These arise naturally in eigenvalue computation because the QR iteration is less expensive for matrices in Hessenberg form. Such matrices (usually in the equivalent lower Hessenberg form) also occur frequently in number-theoretic studies: see for example~\cite{cahill2002fibonacci}, \cite{lettington2011fleck} and~\cite{kilic2010generalized}. 
%It has also been shown that every $n \times n$ non-derogatory matrix is similar to a unique unit upper Hessenberg Toeplitz matrix~\cite{mackey1999every}. We refer to the work of Boyoga, B\"ottcher, and Grudsky \cite{bogoya2012asymptotics, bogoya2012eigenvectors, bogoya2011eigenvalues} for asymptotic results on upper Hessenberg Toeplitz matrices, and to \cite{ching1993maximum}, \cite{inselberg1978determinants}, and \cite{merca2013note} for their results on determinants.

We begin our study in this paper by considering determinants of upper Hessenberg Bohemians. We use two recursive formulae for the characteristic polynomials of upper Hessenberg matrices. During the course of our experimental symbolic and numeric computations, which we report on elsewhere, we encountered ``maximal polynomial height'' characteristic polynomials when the matrices were not only upper Hessenberg, but Toeplitz (that is, $h_{i,j}$ constant along diagonals $j-i = k$). Further restrictions to this class allowed identification of key results including explicit formulae for the characteristic polynomials of maximal height. In what follows, we lay out definitions and prove several facts of interest about characteristic polynomials and their respective height for these families.

\section{Notation}
In what follows, we present some results on upper Hessenberg Bohemians of the form
\begin{equation}
    \H_n =
    \renewcommand{\arraystretch}{1.3}
    \begin{bmatrix}
        h_{1,1} & h_{1,2}    & h_{1,3}     & \cdots & h_{1,n}\\
        s_{1}   & h_{2,2}    & h_{2,3}    & \cdots & h_{2,n}\\
        0     & s_{2}    & h_{3,3}    & \cdots & h_{3,n}\\
        \vdots & \ddots & \ddots & \ddots & \vdots\\
        0      & \cdots & 0       & s_{{n-1}}    & h_{n,n}
    \end{bmatrix}
\end{equation}
with $s_k = e^{i\theta_k}$ and usually $s_k \in \{-1, {+1}\}$ (we do not allow zero $s_k$ entries, because that reduces the problem to smaller ones) and $h_{i,j} \in P$ for $1 \le i \le j \le n$. We denote the characteristic polynomial $Q_n(z) \equiv \det (z \I - \H_n)$.

The \textsl{height} of a matrix $\A$, written $\mathrm{height}(\A) := ||\mathrm{vec}(\A) ||_{\infty}$, is the largest absolute value of any entry in $\A$.  The \textsl{characteristic height} of a matrix is the height of its characteristic polynomial $\det(\lambda\I - \A) = a_0 + a_1\lambda + \cdots + \lambda^n$: that is, the largest absolute value of any coefficient of the characteristic polynomial (expressed in the monomial basis). Equivalently, this is $\| [a_0, a_1, \ldots, a_{n-1}, 1]\|_\infty$.

\begin{definition}
    The set of all $n \times n$ upper Hessenberg Bohemians with upper triangle population $P$ and subdiagonal population from a discrete set of roots of unity, say $s\in \{e^{i\theta_{k}}\}$ where $\{\theta_{k}\}$ is some finite set of angles\footnote{While we mostly use just $\{0\}$ and $\{\pi\}$ in this paper, there are cases when we wish to use other angles, or several together.}, is called \HH{\theta_k}{n}{P}. In particular, \HH{0}{n}{P} is the set of all $n \times n$ upper Hessenberg Bohemians with upper triangle entries from $P$ and subdiagonal entries equal to $1$ and \HH{\pi}{n}{P} is when the subdiagonals entries are $-1$.
\end{definition}

It will often be true that the average value of a population will be zero.  In that case, matrices with trace zero will be common.  It is a useful oversimplification to look in that case at matrices whose diagonal is exactly zero. 

\begin{definition}
    For a population $P$ such that $0 \in P$, let $\ZH{\theta_k}{n}{P}$ be the subset of $\HH{\theta_k}{n}{P}$ where the main diagonal entries are fixed at 0.
\end{definition}

% ============================================================================ %
% Experimental Results                                                         %
% ============================================================================ %
% \section{Results of Experiments}
% \input{ExperimentalResults}

% ============================================================================ %
% Upper Hessenberg Matrices                                                    %
% ============================================================================ %
\section{Upper Hessenberg Matrices}

We begin with a recurrence relation for the characteristic polynomial $Q_n(z) = \det (z \I - \H_n)$ for $\H_n \in \HH{\theta_k}{n}{P}$ where $s = \exp(i\theta_k)$. Later we will specialize the population $P$ to contain only zero and numbers of unit magnitude, usually $\{-1,\, 0,\, {+1}\}$.
% ---------------------------------------------------------------------------- %
% Characteristic Polynomial Recurrence 1                                       %
% ---------------------------------------------------------------------------- %
\begin{theorem}
    \label{thm:charPolyRec1}
    \begin{equation}
        \label{eqn:thm1}
        Q_n(z) = zQ_{n-1}(z) - \sum_{k=1}^{n} \left(\prod_{j=n-k+1}^{n-1}s_j\right) h_{n-k+1,n} Q_{n-k}(z)
    \end{equation}
    with the conventions that an empty product is $1$,  $Q_0(z) = 1$ and $\H_0 = [\,]$, that is the empty matrix.
\end{theorem}
\begin{proof}
Theorems equivalent to this theorem are proved in several places, for instance~\cite{cahill2002fibonacci}. One strategy is Laplace expansion about the final column.
\end{proof}
\begin{theorem}
    \label{thm:charPolyRec2}
    Expanding $Q_n(z)$ as
    \begin{equation}
        Q_n(z) = q_{n,n} z^n + q_{n,n-1} z^{n-1} + \cdots + q_{n,0},
    \end{equation}
    we can express the coefficients recursively by
    \begin{subequations}
    \label{eqn:thm2_eqn}
    \begin{align}
        q_{n,n} &= 1,\\
        q_{n,j} &= q_{n-1,j-1} - \sum_{k=1}^{n-j} \left(\prod_{j=n-k+1}^{n-1}s_j\right) h_{n-k+1,n} q_{n-k,j} \quad\text{for}\quad 1 \le j \le n-1,\\
        q_{n,0} &= -\sum_{k=1}^n \left(\prod_{j=n-k+1}^{n-1}s_j\right) h_{n-k+1,n} q_{n-k,0} \quad\text{for}\quad n>0,\quad\text{and}\\
        q_{0,0} &= 1\>.
    \end{align}
    \end{subequations}
\end{theorem}
\begin{proof}
    % Assume
    % \begin{equation}
    %     Q_n(z) = \sum_{i=0}^n q_{n,i} z^i \>.
    % \end{equation}
    By Theorem~\ref{thm:charPolyRec1}
    \begin{equation}
        Q_n(z) = zQ_{n-1}(z) - \sum_{k=1}^n \left(\prod_{j=n-k+1}^{n-1}s_j\right) h_{n-k+1,n} Q_{n-k}(z) \>.
    \end{equation}
    The first term can be written
    \begin{align}
        zQ_{n-1}(z) 
        % &= z \left [z^{n-1} + q_{n-1,n-2} z^{n-2} + \cdots + q_{n-1, 0} \right ]\\
        % &= z \left [z^{n-1} + \sum_{j=0}^{n-2} q_{n-1,j} z^j \right ]\\
        % &= z^n + \sum_{j=0}^{n-2} q_{n-1,j} z^{j+1}\\
        &= z^n + \sum_{j=1}^{n-1} q_{n-1,j-1} z^{j} \quad \text{(direct from hypothesis)}
    \end{align}
    and the second term is
    \begin{align}
        \left(\prod_{j=n-k+1}^{n-1}s_j\right) h_{n-k+1,n} Q_{n-k}(z) 
        % &= s^{k-1} h_{n-k+1,n} \left [q_{n-k,n-k} z^{n-k} + q_{n-k,n-k-1}z^{n-k-1} + \cdots + q_{n-k,0}\right ] \notag\\
        &= \left(\prod_{j=n-k+1}^{n-1}s_j\right) h_{n-k+1,n} \sum_{j=0}^{n-k} q_{n-k,j} z^j \>.
    \end{align}
    Therefore,
    \begin{align*}
        Q_n(z) &= z^n + \sum_{j=1}^{n-1} q_{n-1,j-1} z^j - \sum_{k=1}^n \left(\prod_{j=n-k+1}^{n-1}s_j\right) h_{n-k+1,n} \sum_{j=0}^{n-k} q_{n-k,j} z^j\\
        &= z^n + \sum_{j=1}^{n-1} q_{n-1,j-1} z^j - \sum_{\ell=0}^{n-1} \left ( \sum_{k=1}^{n-j} \left(\prod_{j=n-k+1}^{n-1}s_j\right) h_{n-k+1,n} q_{n-k,j} \right ) z^\ell\\
        &= z^n + \sum_{\ell=1}^{n-1} \left ( q_{n-1,j-1} - \sum_{k=1}^{n-j} \left(\prod_{j=n-k+1}^{n-1}s_j\right) h_{n-k+1,n} q_{n-k,j} \right ) z^\ell \nonumber\\
        &\qquad - \sum_{k=1}^n \left(\prod_{j=n-k+1}^{n-1}s_j\right) h_{n-k+1,n} q_{n-k,0} \>.
    \end{align*}
\end{proof}

\begin{proposition}
All matrices in $\HH{\theta_k}{n}{P}$ are non-derogatory\footnote{A non-derogatory matrix is a matrix with geometric multiplicity $1$ for every eigenvalue. This implies that its characteristic polynomial and minimal polynomial coincide (up to a factor of $\pm 1$).  This is Fact 38 in~\cite{DeAlba:HLA:2013}}.
\end{proposition}
\begin{proof}
% Let $\H \in \HH{\theta_k}{n}{P}$. Because $\H$ is upper Hessenberg
% \begin{equation}
%     \H_{i,j}^k =
%     \begin{cases}
%     f_{i,j,k} &\quad\text{for}\quad i < j+k\\
%     s^k &\quad\text{for}\quad i = j + k\\
%     0 &\quad\text{for}\quad i > j+k\\
%     \end{cases}
% \end{equation}
% for $0 \le k \le n-1$ where $f_{i,j,k}$ are some functions of the entries of $\H$. Let
% \begin{equation}
%     \A = r(\H) = \sum_{k=0}^{n-1} c_{k} \H^k  = \mathbf{0}\> .
% \end{equation}
% We find $\A_{n,1} = s^{n-1} c_{n-1} = 0$ and therefore $c_{n-1} = 0$. Continuing recursively for $k$ from $n-2$ to 1 we find $\A_{k+j,j} = s^k c_k = 0$ for $1 \le j \le n-k$ and therefore $c_k = 0$ (since $c_j = 0$ for $j > k$) for $1 \le k \le n-1$. We have $\A = c_0 \H^0 = \mathbf{0}$ and hence $c_0 = 0$. Thus, no non-zero polynomial of degree less than $n$ exists that satisfies $r(\H) = \mathbf{0}$. Therefore, the minimal degree non-zero polynomial that satisfies $r(\H) = \mathbf{0}$ is the characteristic polynomial of $\H$.
This follows from the irreducibility (as a matrix) of the upper Hessenberg matrix $\H$, because the $k$th power of $\H$ has nonzero $k$th subdiagonal.
\end{proof}

% ---------------------------------------------------------------------------- %
% Characteristic Height Definition                                             %
% ---------------------------------------------------------------------------- %
\begin{definition}
    The \textit{characteristic height} of a matrix is the height of its
    characteristic polynomial.
\end{definition}
\begin{remark}
    The height of a polynomial is in fact a norm, namely the infinity norm of the vector of coefficients. In number theory of polynomials, the $1$-norm of the vector of coefficients is also used, and it is called the \textsl{length} of a polynomial in that literature.  We do not use \textsl{length} here.
\end{remark}

% ---------------------------------------------------------------------------- %
% A, -A same height proposition                                                %
% ---------------------------------------------------------------------------- %
\begin{proposition}
    \label{prop:negativeheight}
    For any matrix $\A$, $-\A$ has the same characteristic height.
\end{proposition}
\begin{proposition}
    \label{prop:maxheight}
    For populations $P$ with maximal height $1$, the maximal characteristic height of $\H_n \in \mathcal{H}_{\{\theta_k\}}^{n \times n}(P)$ occurs when
    $\left(\prod_{j=n-k+1}^{n-1}s_j\right) h_{i,i+k-1} = -1$ for $1 \le i \le n-k+1$ and $1 \le k \le n$.
\end{proposition}
\begin{proof}
    Since $|s_k|=1$ and $|h_{i,j}| \le 1$,  $\max |\left(\prod_{j=n-k+1}^{n-1}s_j\right) h_{i,i+k-1}| \le 1$. Suppose $\left(\prod_{j=n-k+1}^{n-1}s_j\right) h_{i,i+k-1} = -1$. By Theorem~\ref{thm:charPolyRec2}
    \begin{align}
        q_{n,0} = -\sum_{k=1}^n \left(\prod_{j=n-k+1}^{n-1}s_j\right) h_{n-k+1,n} q_{n-k,0}
                = \sum_{k=1}^n q_{n-k,0} \label{eq:prop1_3_MaxHeight}
    \end{align}
    and
    \begin{align}
        q_{n,j} = q_{n-1,j-1} - \sum_{k=1}^{n-j} \left(\prod_{j=n-k+1}^{n-1}s_j\right) h_{n-k+1,n} q_{n-k,j}
                = q_{n-1,j-1} + \sum_{k=1}^{n-j} q_{n-k,j} \>. \label{eq:prop1_2_MaxHeight}
    \end{align}
    Since $q_{0,0} = 1$, and equations~\eqref{eq:prop1_3_MaxHeight} and \eqref{eq:prop1_2_MaxHeight} are independent of $s$ and $h_{i,j}$, all $q_{n,j}$ must be positive and the maximum characteristic height is attained.
\end{proof}
\begin{remark}
    When all $s_j = 1$ ($\theta = 0$) and $h_{i,j} = -1$ for all $1 \le i \le j \le n$, $\H_n$ attains maximal characteristic height. By Proposition~\ref{prop:negativeheight}, $s = -1$ ($\theta = \pi$) and $h_{i,j} = 1$ will also attain maximal characteristic height. Both of these cases correspond to upper Hessenberg matrices with a Toeplitz structure as we explore in further detail in Section~\ref{sec:UHTM}.
\end{remark}

\begin{definition}
    We say that $P$ is \textbf{invariant under multiplication} by a fixed unit $e^{i\theta}$ if $e^{i\theta}P = P$; that is, each entry of $P$, say $p$, is such that $e^{i\theta}p$ is also in $P$. Note that invariance with respect to $e^{i\theta}$ implies invariance with respect to $e^{-i\theta}$.
    % For instance, $\{-1, 0, {+1}\}$ is invariant under multiplication by $-1$.
\end{definition}

% ---------------------------------------------------------------------------- %
% Property where subdiagonal entries are all -1 or 1                           %
% ---------------------------------------------------------------------------- %
\begin{theorem}
    Suppose $\H_n \in \HH{\theta_k}{n}{P}$ and $P$ is invariant under multiplication by each $e^{i\theta_{k}}$ and by $-e^{i\theta_{k}}$. Then $\H_n$ is similar to a matrix in $\HH{\pi}{n}{P}$, and similar to a matrix in $\HH{0}{n}{P}$.
\end{theorem}

\begin{proof}
    We use induction. The case $n=1$ is vacuously upper Hessenberg.  For completeness, we have
    \begin{equation*}
        \begin{bmatrix}
            e^{i\theta_{k}}
        \end{bmatrix}
        \begin{bmatrix}
            h_{11}
        \end{bmatrix}
        \begin{bmatrix}
            e^{-i\theta_{k}}
        \end{bmatrix}
        =
        \begin{bmatrix}
            h_{11}
        \end{bmatrix}
        \in \HH{\theta_k}{1}{P} \>.
    \end{equation*}
    For $n > 1$, partition the matrix as
    \begin{equation*}
        \left[
        \begin{array}{cccc}
            h_{11} & h_{12} \cdots & h_{1n} \\
            s & \multicolumn{3}{c}{\multirow{3}{*}{$\H_{n-1}$}} \\
            & & & \\
            & & &
        \end{array}
        \right]
    \end{equation*}
    where $s = e^{i\theta_{k}}$ for some $\theta_{k}$. Then conjugate by
    \begin{multline*}
        \begin{bmatrix}
            1 & & \\
            & e^{-i\theta_{k}} & \\
            & & \I_{n-2}
        \end{bmatrix}
        \left[
        \begin{array}{cccc}
            h_{11} & h_{12} \cdots & h_{1n} \\
            s & \multicolumn{3}{c}{\multirow{3}{*}{$\H_{n-1}$}} \\
            & & & \\
            & & &
        \end{array}
        \right]
        \begin{bmatrix}
            1 & & \\
            & e^{-i\theta_{k}} & \\
            & & \I_{n-2}
        \end{bmatrix} ^{-1} \\
        =
        \left[
        \begin{array}{ccc}
            h_{11} & e^{i\theta_{k}}h_{12} & \cdots \\
            1 & \multicolumn{2}{c}{\multirow{2}{*}{$\tilde{\H}_{n-1}$}} \\
            & &
        \end{array}
        \right] \>.
    \end{multline*}
    Clearly $\tilde{\H}_{n-1}$ is in $\HH{\theta_k}{n-1}{P}$. By induction the proof is complete. The proof for $-e^{i\theta k}$ and $\HH{\pi}{n}{P}$
\end{proof}

\begin{remark}
    For clarity, consider the case $n = 2$:
    \begin{equation}
         \H =
        \begin{bmatrix}
            a & b \\
            s & c
        \end{bmatrix} \>,
    \end{equation}
    where $a, b, c \in P$ and $s = e^{i\theta_{k}}$. Then, the following similarity transforms reduce the problem to one in \HH{0}{2}{P} and one in \HH{\pi}{2}{P}.
    \begin{align}
        \begin{bmatrix}
            1 & 0 \\
            0 & e^{-i\theta_{k}}
        \end{bmatrix}
        \H
        \begin{bmatrix}
            1 & 0 \\
            0 & e^{i\theta_{k}}
        \end{bmatrix}
        & =
        \begin{bmatrix}
            a & be^{i\theta_{k}} \\
            1 & c
        \end{bmatrix} \\
        \begin{bmatrix}
            1 & 0 \\
            0 & -e^{-i\theta_{k}}
        \end{bmatrix}
        \H
        \begin{bmatrix}
            1 & 0 \\
            0 & -e^{i\theta_{k}}
        \end{bmatrix}
        &=
        \begin{bmatrix}
            a & -be^{i\theta_{k}} \\
            -1 & c
        \end{bmatrix} \>.
    \end{align}
\end{remark}

% ============================================================================ %
% Zero diagonal and $-1$ subdiagonal for Upper Hessenberg Matrices             %
% ============================================================================ %
\subsection{Zero Diagonal Upper Hessenberg Matrices}

% Suppose \linebreak $\mathbf{M}_{n} \in \ZH{0}{n}{P}$, where $P = \{0, w_{1}, \ldots, w_{m}\}$ for some fixed positive integer $m$, and each $\left|w_{j}\right| = 1$. The simplest example is $\{0, 1\}$ corresponding to $m=1$, but $m=2$ when $P = \{-1, 0, 1\}$ is interesting.
In this section we make the simplifying assumption that the diagonal of the matrix is zero.  This amounts, in the event that a population is symmetric about zero, to looking at the distribution about the average eigenvalue.  This is precisely true in the Toeplitz case, but only indicative otherwise.

In this section we also investigate just when the matrices can be \textsl{normal} (that is, commute with their Hermitian transpose).  Normal matrices have several nice properties; matrices that are not normal have interesting pseudospectra~\cite{Embree:HLA:2013}.  We find that for upper Hessenberg matrices, ``normality'' is atypical.  In retrospect, this could have been expected.
\begin{theorem}
\label{thm:zero_diag_UH}
Let $\A_n \in \ZH{0}{n}{P}$ for $P = \{0, w_1, \ldots, w_m\}$ for some fixed positive integer $m$ and each $|w_j| = 1$. If $\A_{n}$ is normal, i.e.~$\A_{n}^{*}\A_{n} = \A_{n}\A_{n}^{*}$, then for $n \geq 3$, $\A_{n}$ is $w_{j}$-skew symmetric for some fixed $1 \leq j \leq m$ or $w_{j}$-skew circulant. These $2m$ matrices ($m$ symmetric/$w_j$-skew symmetric, and $m$ $w_j$-skew circulant matrices) are the only normal matrices in $\ZH{0}{n}{P}$. (For $n=1$, this is only $\left[ 0 \right]$; for $n=2$, the symmetric and circulant cases coalesce, so that there are only $m$ such matrices.)
\end{theorem}
\begin{proof}
To prove this theorem, we establish a sequence of lemmas. First, we partition $\A_{n}$. Put
\begin{equation}
    \A_{n} =
    \begin{bmatrix}
        0 & \T^{*} \\
        \mathbf{e} & \A_{n-1}
    \end{bmatrix}
\end{equation}
where
\begin{equation}
    \mathbf{e}^{*} =
    \begin{bmatrix}
        1 & 0 & \cdots & 0
    \end{bmatrix}
\end{equation}
and
\begin{equation}
    \T^{*} =
    \begin{bmatrix}
        t_{12} & t_{13} & \cdots & t_{1n}
    \end{bmatrix} \>.
\end{equation}
Then the conditions of normality are
\begin{equation}
    \A_{n}\A_{n}^{*} =
    \begin{bmatrix}
        \T^{*}\T & \T^{*}\A_{n-1}^{*} \\
        \A_{n-1} & \mathbf{e}\mathbf{e}^{*} + \A_{n-1}\A_{n-1}^{*}
    \end{bmatrix}
\end{equation}
must equal
\begin{equation}
    \A_{n}^{*}\A_{n} =
    \begin{bmatrix}
        1 & \mathbf{e}^{*}\A_{n-1} \\
        \A_{n-1}^{*}\mathbf{e} & \T\T^{*} + \A_{n-1}^{*}\A_{n-1}
    \end{bmatrix} \>.
\end{equation}
\end{proof}

\begin{lemma}
\label{lemma:nonzero_Mn}
The first row of $\A_{n}$ contains exactly one nonzero element, say $\tau$ in position $j$ $(2 \leq j \leq n)$.
\end{lemma}

\begin{proof}
\begin{equation}
    \T^{*}\T = \sum_{j=2}^{n} \left|t_{ij}\right|^{2} = 1
\end{equation}
from the upper left corner. Since each nonzero element of $P$ has magnitude $1$, exactly one entry must be nonzero.
\end{proof}

\begin{lemma}
If $\A_{n-1}$ is normal then $\T = \tau\mathbf{e}$ and $\A_{n}$ is $\tau$-skew symmetric.
\end{lemma}

\begin{proof}
If $\A_{n-1}$ is normal, then $\T\T^{*} + \A_{n-1}^{*}\A_{n-1}$ being equal to $\mathbf{e}\mathbf{e}^{*} + \A_{n-1}\A_{n-1}^{*}$ implies $\T\T^{*} = \mathbf{e}\mathbf{e}^{*}$ so that $\T^{*} = \begin{bmatrix}\tau^{*} & 0 & \cdots & 0\end{bmatrix}$ for some $\tau$ with $\left|\tau\right| = 1$. Then 
\begin{equation}
\T^{*}\A_{n-1}^{*} = \mathbf{e}^{*}\A_{n-1} \Rightarrow \tau^{*}\begin{bmatrix} 1 & 0 & \cdots & 0 \end{bmatrix}\A_{n-1}^{*} = \mathbf{e}^{*}\A_{n-1}^{*}
\end{equation}
and this says $\tau^{*}$ times the first row of $\A_{n-1}^{*}$ is the first row of $\A_{n-1}$.

But the first row of $\A_{n-1}^{*}$ is $\begin{bmatrix}0 & 1 & 0 & \cdots & 0\end{bmatrix}$ because $\A_{n-1}$ is upper Hessenberg with zero diagonal. Thus the first row of $\A_{n-1}$ is $\begin{bmatrix}0 & \tau^{*} & 0 & \cdots & 0\end{bmatrix}$. Thus
\begin{equation}
    \A_{n} =
    \left[
        \begin{array}{c|c|c}
            0 & \tau^{*} & \\
            \hline
            1 & 0 & \begin{array}{ccc}\tau^{*} & &\end{array} \\
            \hline
            & \begin{array}{c}1 \\ \\\end{array} & \A_{n-2}
        \end{array}
    \right]
    \quad \text{(remember $n \geq 3$)}
\end{equation}
and
\begin{equation}
    \A_{n-1}=
    \left[
        \begin{array}{cc}
            0 & \begin{array}{ccc} \tau^{*} & & \end{array} \\
            \begin{array}{c} 1 \\ \\ \end{array} & \A_{n-2}
        \end{array}
    \right]
\end{equation}
is normal. Because $\A_{n-1}$ is normal, and 
\begin{equation}
    \A_{n-1}^{*} = 
    \left[
        \begin{array}{ccc}
            0 & 1 & \\
            \tau & 0 & \begin{array}{ccc}1 & & \end{array} \\
            & \begin{array}{c}\tau \\ \\ \end{array} & \A_{n-2}^{*}
        \end{array}
    \right]
\end{equation}
we have $\A_{n-1}^{*}\A_{n-1} = \A_{n-1}\A_{n-1}^{*}$ or
\begin{multline*}
    \left[
        \begin{array}{ccc}
            0 & 1 & \\
            \tau & 0 & \begin{array}{ccc} 1 & & \end{array} \\
            & \begin{array}{c}1 \\ \\ \end{array} & \A_{n-2}
        \end{array}
    \right]
    \left[
        \begin{array}{ccc}
            0 & \tau^{*} & \\
            1 & 0 & \begin{array}{ccc}\tau^{*} & & \end{array} \\
            & \begin{array}{c}1 \\ \\ \end{array} & \A_{n-2}
        \end{array}
    \right] \\
    =
    \left[
        \begin{array}{ccl}
            1 & 0 & \tau^{*} \\
            0 & 2 & \mathbf{e}_{n-2}^{*}\A_{n-2} \\
            \begin{array}{c}\tau \\ \\ \end{array} & \tau\A_{n-2}^{+}\mathbf{e}_{n-2} & \mathbf{e}\mathbf{e}^{*} + \A_{n-2}^{*}\A_{n-2}
        \end{array}
    \right]
\end{multline*}
must equal
\begin{multline*}
    \left[
        \begin{array}{ccc}
            0 & \tau^{*} & \\
            1 & 0 & \begin{array}{ccc}\tau^{*} & & \end{array} \\
            & \begin{array}{c}1 \\ \\ \end{array} & \A_{n-2}
        \end{array}
    \right]
    \left[
        \begin{array}{ccc}
            0 & 1 & \\
            \tau & 0 & \begin{array}{ccc} 1 & & \end{array} \\
            & \begin{array}{c}1 \\ \\ \end{array} & \A_{n-2}
        \end{array}
    \right] \\
    =
    \left[
        \begin{array}{ccl}
            1 & 0 & \tau^{*} \\
            0 & 2 & \mathbf{e}_{n-2}^{*}\A_{n-2} \\
            \begin{array}{c}\tau \\ \\ \end{array} & \tau\A_{n-2}^{+}\mathbf{e}_{n-2} & \mathbf{e}\mathbf{e}^{*} + \A_{n-2}^{*}\A_{n-2}
        \end{array}
    \right] \>.
\end{multline*}
The lower left block gives $\mathbf{e}\mathbf{e}^{*} + \A_{n-2}\A_{n-2}^{*} = \mathbf{e}\mathbf{e}^{*} + \A_{n-2}^{*}\A_{n-2}$ so $\A_{n-2}$ must also be normal.

At this point, we see the outline of an induction:
\begin{equation}
    \A_{n} =
    \left[
        \begin{array}{c|c}
            0 & \begin{array}{ccc}\tau^{*} & & \end{array} \\
            \hline
            \begin{array}{c}1 \\ \\ \end{array} & \A_{n-1}
        \end{array}
    \right]
\end{equation}
being normal with $\A_{n-1}$ also being normal implies that
\begin{equation}
    \A_{n-1} =
    \left[
        \begin{array}{c|c}
            0 & \begin{array}{ccc}\tau^{*} & & \end{array} \\
            \hline
            \begin{array}{c}1 \\ \\ \end{array} & \A_{n-2}
        \end{array}
    \right]
\end{equation}
where $\A_{n-2}$ is normal. Explicit computation of the $n=3$ case shows the induction terminates.
\end{proof}

We now consider the harder case where
\begin{equation}
    \A_{n} =
    \begin{bmatrix}
        0 & \T^{*} \\
        \mathbf{e}_{n-1} & \A_{n-1}
    \end{bmatrix}
\end{equation}
but where $\A_{n-1}$ is not itself normal. From Lemma \ref{lemma:nonzero_Mn} we know that $\T^{*}$ has only one nonzero element; call it $\tau^{*}$ as before. Then
\begin{equation}
    \T\T^{*} =
    \begin{bmatrix}
        0 & & & & & & \\
        & \ddots & & & & & \\
        & & 0 & & & & \\
        & & & 1 & & & \\
        & & & & 0 & & \\
        & & & & & \ddots & \\
        & & & & & & 0
    \end{bmatrix}
\end{equation}
while
\begin{equation}
    \mathbf{e}\mathbf{e}^{*} =
    \begin{bmatrix}
        1 & & & \\
        & 0 & & \\
        & & \ddots & \\
        & & & 0
    \end{bmatrix} \>,
\end{equation}
and we may assume that the $1$ in $\T\T^{*}$ does not occur in the first row and column (else we are in the previous case, and $\A_{n-1}$ will be normal). Here 
\begin{equation}
    \label{eq:Mnormal}
    \A_{n-1}\A_{n-1}^{*} - \A_{n-1}^{*}\A_{n-1} = \T\T^{*} - \mathbf{e}\mathbf{e}^{*} = 
    \begin{bmatrix}
        -1 & & & & & & & \\
        & 0 & & & & & & \\
        & & \ddots & & & & & \\
        & & & 0 & & & & \\
        & & & & 1 & & & \\
        & & & & & 0 & & \\
        & & & & & & \ddots & \\
        & & & & & & & 0
    \end{bmatrix}
\end{equation}
is the departure of $\A_{n-1}$ from normality. We will establish that in fact 
\begin{equation}
    \T^{*} =
    \begin{bmatrix}
        0 & 0 & 0 & \cdots & 0 & \tau^{*}
    \end{bmatrix}
\end{equation}
and that
\begin{equation}
    \A_{n-1} =
    \begin{bmatrix}
        0 & & & & \\
        1 & 0 & & & \\
        & 1 & 0 & & \\
        & & \ddots & \ddots & \\
        & & & 1 & 0
    \end{bmatrix} \>;
\end{equation}
that is, the nonzero element can only occur in the last place. Notice that the upper left corner of equation \eqref{eq:Mnormal} is, if the top row of $\A_{n-1}$ is $\begin{bmatrix}0 & a_{1,2} & a_{1,3} & \cdots a_{1,n-1}\end{bmatrix}$,
\begin{equation}
    \sum_{j=2}^{n-1} \left|a_{1,j}\right|^{2} - 1 = -1 \>.
\end{equation}
Therefore, all $a_{1,j} = 0$ and the first row of $\A_{n-1}$ must be zero: i.e.
\begin{equation}
    \A_{n-1} =
    \begin{bmatrix}
        0 & 0 & 0 & \cdots & 0 \\
        1 & 0 & a_{3,3} & \cdots & a_{2, n-1} \\
        & 1 & 0 & \ddots & \vdots \\
        & & \ddots & \ddots & a_{n-2, n-1} \\
        & & & 1 & 0 
    \end{bmatrix}
\end{equation}
Then,
\begin{equation}
    \A_{n-1}\T = \A_{n-1}^{*}\mathbf{e} = 
    \begin{bmatrix}
        0 & 1 & & \\
        0 & 0 & \ddots \\
        \vdots & & \ddots & 1 \\
        0 & \cdots & \cdots & 0 
    \end{bmatrix}
    \begin{bmatrix}
        1 \\
        0 \\
        \vdots \\
        0
    \end{bmatrix}
    =
    \begin{bmatrix}
        0 \\
        0 \\
        \vdots \\
        0
    \end{bmatrix} \>.
\end{equation}
If
\begin{equation}
    \T =
    % \begin{bmatrix}
    %     0 \\
    %     0 \\
    %     \vdots \\
    %     0 \\
    %     \tau \\
    %     0 \\
    %     \vdots \\
    %     0
    % \end{bmatrix}\>,
    [ 0, 0, \ldots, 0, \tau^*, 0, \ldots, 0]^*
\end{equation}
then
\begin{equation}
    \A_{n-1}\T =
    % \begin{bmatrix}
    %     0 \\
    %     \tau a_{2,j} \\
    %     \vdots \\
    %     \tau a_{j-1, j} \\
    %     0 \\
    %     \tau \\
    %     0 \\
    %     \vdots \\
    %     0
    % \end{bmatrix}\>,
    [0, \tau^*a_{2,j}^*, \ldots, \tau^*a_{j-1,j}, 0, \ldots, 0]^*
\end{equation}
which is impossible unless $j=n$ (when the $\tau$ term is not present). Therefore,
\begin{equation}
    \A_{n-1} =
    \begin{bmatrix}
        0 & 0 & \cdots & 0 & 0 \\
        1 & x & \cdots & x & 0 \\
        & 1 & \ddots & \vdots & \vdots \\
        & & \ddots & x & 0 \\
        & & & 1 & 0
    \end{bmatrix}
    =
    \begin{bmatrix}
        0 & 0 \\
        \U & 0
    \end{bmatrix} \>,
\end{equation}
and
\begin{equation}
    \A_{n-1}\A_{n-1}^{*} - \A_{n-1}^{*}\A_{n-1} =
    \begin{bmatrix}
        -1 & & & & \\
        & 0 & & & \\
        & & \ddots & & \\
        & & & 0 & \\
        & & & & 1
    \end{bmatrix} \>.
\end{equation}
Since
\begin{equation}
    \A_{n-1}^{*} =
    \begin{bmatrix}
        0 & \U^{*} \\
        0 & 0
    \end{bmatrix}
\end{equation}
and
\begin{equation}
    \A_{n-1}\A_{n-1}^{*} =
    \begin{bmatrix}
        0 & 0 \\
        0 & \U\U^{*}
    \end{bmatrix}
\end{equation}
and
\begin{equation}
    \A_{n-1}^{*}\A_{n-1} =
    \begin{bmatrix}
        \U^{*}\U & 0 \\
        0 & 0
    \end{bmatrix} \>,
\end{equation}
\begin{equation}
    % \A_{n-1}\A_{n-1}^{*} - \A_{n-1}^{*}\A_{n-1} =    
    \begin{bmatrix}
        0 & 0 \\
        0 & \U\U^{*}
    \end{bmatrix}
    -
    \begin{bmatrix}
        \U^{*}\U & 0 \\
        0 & 0
    \end{bmatrix}
\end{equation}
must be diagonal. Therefore, the first row of $\U\U^{*}$ must be zero except for the first element.

\begin{remark}
    For $n = 4$, and $P = \{0, i, -i\}$ ($m = 2$) the following 4 matrices are normal:
    
    \begin{center}
    {\renewcommand{\arraystretch}{1}
    \begin{tabular}{*1{>{\centering\arraybackslash}p{.05\textwidth}}*2{>{\centering\arraybackslash}p{.3\textwidth}}}
    
    \toprule
    $w_j$ & $w_j$-skew symmetric & $w_j$-skew circulant\\ \midrule
    \[i\] &
    \[\begin{bmatrix}
        0 & i & 0 & 0\\
        1 & 0 & i & 0\\
          & 1 & 0 & i\\
          &   & 1 & 0
    \end{bmatrix}\] & 
    \[\begin{bmatrix}
        0 & 0 & 0 & i\\
        1 & 0 & 0 & 0\\
          & 1 & 0 & 0\\
          &   & 1 & 0
    \end{bmatrix}\]\\
    \[-i\] &
    \[\begin{bmatrix}
        0 & -i & 0 & 0\\
        1 & 0 & -i & 0\\
          & 1 & 0 & -i\\
          &   & 1 & 0
    \end{bmatrix}\] & 
    \[\begin{bmatrix}
        0 & 0 & 0 & -i\\
        1 & 0 & 0 & 0\\
          & 1 & 0 & 0\\
          &   & 1 & 0
    \end{bmatrix}\]\\ \bottomrule
    \end{tabular}}
    \end{center}
\end{remark}

% ============================================================================ %
% Stable Matrices                                                              %
% ============================================================================ %
\subsection{Stable Matrices}
An important question in dynamical systems (either continuous or discrete), especially in models in mathematical biology where we find the notion of \textsl{sign-stability}, is whether or not the system is \textsl{stable}. That is, does the solution of the system (or of perturbations to the system) ultimately decay to zero.  The theory of eigenvalues is classically connected to this question.  For Bohemians, the natural version of this is to ask what is the probability that the chosen Bohemian is stable?  We investigate this question in this section.

\subsection{Type I Stable Matrices}
A \textsl{Type I stable matrix} $\A$
is a matrix with all of its eigenvalues strictly in the left half plane: if $\lambda$ is an eigenvalue of~$A$ then $\Re(\lambda) < 0$. This nomenclature comes from differential equations, in that all solutions of the linear system of ODEs $dy/dt = \A y$ will ultimately decay as $t \to \infty$ if $\A$ is a type I stable matrix. 

If the matrix $\A$ is not \textsl{normal}, then \textsl{pseudospectra} can play a role, in that even though all solutions $y$ must ultimately decay, they might first grow large.  See~\cite{Embree:HLA:2013} for details.

By Theorem~\ref{thm:zero_diag_UH}, only $2m$ of the zero diagonal upper Hessenberg matrices with population $P = \{-1, 0, {+1}\}$ are normal, where here $m=2$.  Similarly, when the population is $P=\{0, {+1}\}$ then $m=1$ and only two matrices of every dimension are normal (the symmetric matrix with $1$s on its upper diagonal, and the circulant matrix with a $1$ in the last column of the first row).

\begin{theorem}
\label{thm:not_stable}
If $\mathrm{trace}(\A) = 0$, where $\A \in \mathbb{C}^{n \times n}$, then $\A$ is not type 1 stable.
% No $\A_n \in \ZH{\theta_k}{n}{\mathbb{Z}}$ is Type I stable, for any population $P$.
\end{theorem}

\begin{proof}
Suppose $\A$ has eigenvalues $\{\lambda_{k}\}_{k=1}^{n}$. Then 
\begin{equation}
\sum_{k=1}^{n}\lambda_{k} = \mathrm{trace}(\A_n) = 0\>. 
\end{equation}
Therefore, $\sum_{k=1}^{n}\mathrm{Re}(\lambda_{k}) = 0$. This is $n$ times the average, and so the average is zero.  Since the maximum $\mathrm{Re}(\lambda_{k})$ must be larger than the average, this proves the theorem.
% $\mathrm{Tr}(\M) = 0$, so $\sum_{k=1}^{n}\lambda_{k} = 0$, so $\sum_{k=1}^{n}\mathrm{Re}(\lambda_{k}) = 0$, hence $\overline{\mathrm{Re}(\lambda_{k})} = 0$ and $\mathrm{max}\mathrm{Re}(\lambda_{k}) \geq \overline{\mathrm{Re}(\lambda_{k})} = 0$.
\end{proof}
% The proof of this theorem did not depend on the structure or population.  
\begin{corollary}
    No $\A_{n} \in \ZH{0_n}{n}{P}$ is type 1 stable.
\end{corollary}

\subsection{Type II Stable matrices}
A \textsl{Type II Stable Matrix} $\A$ has all its eigenvalues inside the unit circle.  This class of matrices arises naturally on studying the simple linear recurrence relation $y_{n+1} = \A y_n$.
Fairly obviously, all solutions of this difference equation will ultimately decay to $0$ as $n \to \infty$ if and only if all eigenvalues of $\A$ are inside the unit circle (again, pseudospectra can play a role in the transient behaviour, sometimes significantly).

\begin{theorem}
If $\A \in \mathbb{Z}^{n \times n}$, then it is Type II stable if and only if it is nilpotent.
\end{theorem}
\begin{proof}
Suppose to the contrary that some eigenvalues are not zero.

The determinant of $\A$ must necessarily be an integer.  If the integer is not zero, it is at least $1$ in magnitude.  The product of the eigenvalues is thus at least $1$ in magnitude; hence there must be at least one eigenvalue that is at least $1$ in magnitude.

If the matrix $\A$ has zero determinant but not all eigenvalues zero, then after factoring out $z^m$ for the multiplicity of the zero eigenvalue, the product of the other eigenvalues becomes the constant coefficient (what was the coefficient of $z^m$ in the original).  This coefficient again must be an integer, and again at least one eigenvalue must be at least $1$ in magnitude.

This proves the theorem, by contradiction.
\end{proof}

\begin{corollary}
    If $\A$ is Bohemian with integer population $P$, then it is Type II stable if and only if $\A$ is nilpotent.
\end{corollary}

\begin{remark}
We did not, in fact, use that the matrix came from a Bohemian family; only that its entries were integers.
\end{remark}

% ============================================================================ %
% Upper Hessenberg Toeplitz Matrices                                           %
% ============================================================================ %
\section{Upper Hessenberg Toeplitz Matrices}
\label{sec:UHTM}
Proposition~\ref{prop:maxheight} gives  matrices in \HH{0,\pi}{n}{\{-1, 0, {+1}\}} with maximal characteristic height. We noticed that they are Toeplitz matrices. This motivates our interest in upper Hessenberg Toeplitz matrices.

Consider upper Hessenberg matrices with a Toeplitz structure of the form
\begin{equation}
    \M_n =
    \begin{bmatrix}
        t_1   & t_2    & t_3     & \cdots & t_n\\
        s   & t_1    & t_2     & \cdots & t_{n-1}\\
        0     & s    & t_1     & \cdots & t_{n-2}\\
        \vdots & \ddots & \ddots & \ddots & \vdots\\
        0      & \cdots & 0       & s    & t_1
    \end{bmatrix}
\end{equation}
with $s = e^{i\theta_k}$.  Again we require the matrix to be irreducible, that is, subdiagonal entries cannot be zero.

\begin{definition}
    The set of all $n \times n$ upper Hessenberg Toeplitz Bohemians with upper triangle population $P$ and subdiagonal population from a discrete set of roots of unity, say $s\in \{e^{i\theta_{k}}\}$ where $\{\theta_{k}\}$ is some finite set of angles, is called $\MM{\theta_k}{n}{P}$.
\end{definition}

We will restrict our analysis in this section to those matrices with population $\{-1, 0, {+1}\}$ and subdiagonals fixed at $1$. We will denote this set by 
$$
\MMZ = \MM{0}{n}{\{-1, 0, {+1}\}}\>.
$$
We denote the characteristic polynomial $P_n(z) \equiv \det (z \I - \M_n)$ for $\M_n \in \MMZ$.

% ---------------------------------------------------------------------------- %
% ---------------------------------------------------------------------------- %
\begin{corollary}
    \label{prop:charpolyrec1_UHT}
    The characteristic polynomial recurrence from Theorem~\ref{thm:charPolyRec1} can be written for upper Hessenberg Toeplitz matrices in \MM{0}{n}{P} as
    \begin{equation}
        \label{eqn:thm1_UHT}
        P_n(z) = zP_{n-1}(z) - \sum_{k=1}^n t_k P_{n-k}(z)
    \end{equation}
        with the convention that $P_0(z) = 1$ ($\M_0 = [\,]$, the empty matrix).
\end{corollary}
\begin{proof}
    For a matrix $\M_n \in \MM{0}{n}{P}$, the entries at the $i$th row and the $i+k-1$-th column for $1 \le i \le n-k+1$ (i.e. the $k-1$-th diagonal) are all equal to $t_k$. In equation~\eqref{eqn:thm1}, we can replace $h_{n-k+1,n}$ with $t_k$ ($i=n-k+1$) recovering equation~\eqref{eqn:thm1_UHT}.
\end{proof}

% ---------------------------------------------------------------------------- %
% ---------------------------------------------------------------------------- %
\begin{corollary}
    \label{prop:charpolyrec2_UHT}
    The characteristic polynomial recurrence from Theorem~\ref{thm:charPolyRec2}  can be written for upper Hessenberg Toeplitz matrices in \MM{0}{n}{P} as
    \begin{subequations}
    \label{eqn:thm2_cor_UHT}
    \begin{align}
        p_{n,n} &= 1,\\
        p_{n,j} &= p_{n-1,j-1} - \sum_{k=1}^{n-j} t_k p_{n-k,j} \quad\text{for}\quad 1 \le j \le n-1, \label{eqn:them2_cor_UHT_b}\\
        p_{n,0} &= -\sum_{k=1}^n t_k p_{n-k,0}, \,\text{and} \label{eqn:them2_cor_UHT_c}\\
        p_{0,0} &= 1 \label{eqn:them2_cor_UHT_d}\>.
    \end{align}
    \end{subequations}
\end{corollary}
\begin{proof}
    Performing the same replacement as above (a notational change), we recover equation~\eqref{eqn:thm2_cor_UHT}.
\end{proof}

% ---------------------------------------------------------------------------- %
% Proposition: p_{n,i} is a function of t_j for j <= n-i                       %
% ---------------------------------------------------------------------------- %
\begin{proposition}
    \label{prop:funof}
    $p_{n,i}$ is independent of $t_j$ for $j > n-i$.
\end{proposition}
\begin{proof}
    First, assume $p_{n,\ell}$ is a function of $t_1, \ldots, t_{n-\ell}$ for $\ell = i$ and all $n$. By Proposition~\ref{prop:charpolyrec2_UHT}
    \begin{equation}
        p_{n,\ell} = p_{n-1,\ell-1} - \sum_{k=1}^{n-\ell} t_k p_{n-k,\ell} \>.
    \end{equation}
    Isolating the $p_{n-1,\ell-1}$ term, we have
    \begin{equation}
        p_{n-1,\ell-1} = p_{n,\ell} + \sum_{k=1}^{n-\ell} t_k p_{n-k,\ell}
    \end{equation}
    The first term, $p_{n,\ell}$, is a function of $t_1, \ldots, t_{n-\ell}$. Each term $t_k p_{n-k,\ell}$ in the sum is a function of $t_1, \ldots, t_{n-k-\ell}, t_k$. Taking $k=n-\ell$, we have the sum is a function of $t_1, \ldots, t_{n-\ell}$. Hence, $p_{n-1,\ell-1}$ is a function of $t_1, \ldots, t_{n-1-(\ell-1)} = t_{n-\ell}$.

    When $i = 0$, by Proposition~\ref{prop:charpolyrec2_UHT} we have
    \begin{equation}
        p_{n,0} = -\sum_{k=1}^n t_k p_{n-k,0}
    \end{equation}
    which is a function of $t_1, \ldots, t_n$.
\end{proof}

\begin{theorem}
    The set of distinct characteristic polynomials for all matrices $\M_n \in \MMZ$ has cardinality $\left(\#P\right)^n$, which is the same as the cardinality of $\MMZ$. That is, each matrix in $\MMZ$ has a unique characteristic polynomial.
\end{theorem}
\begin{proof}
    Let
    \begin{equation}
    \A_n =
    \begin{bmatrix}
        a_1   & a_2    & a_3     & \cdots & a_n\\
        1   & a_1    & a_2     & \cdots & a_{n-1}\\
        0     & 1    & a_1     & \cdots & a_{n-2}\\
        \vdots & \ddots & \ddots & \ddots & \vdots\\
        0      & \cdots & 0       & 1    & a_1
    \end{bmatrix}
    \end{equation}
    with $a_k \in P$ for $1 \le k \le n$.
    Let $R_n(z; a_1, \ldots, a_n)$ be the characteristic polynomial of $\A_n$.
    Assume $P_{\ell} = R_{\ell}$ for $\ell < n$.
    By Proposition~\ref{prop:charpolyrec1_UHT}, for $\A_n$ and $\M_n$ to have the same characteristic polynomial we find
    \begin{equation}
    zP_{n-1} - \sum_{k=1}^n t_k P_{n-k} = zR_{n-1} - \sum_{k=1}^n a_k R_{n-k} \>.
    \end{equation}
    Since $P_{\ell} = R_{\ell}$ for all $\ell < n$, and the $\sum_{k=1}^n t_k P_{n-k}$ and $\sum_{k=1}^n t_k R_{n-k}$ terms are polynomials of degree $n-1$ in $z$, we find $P_n = R_n$ only when $t_k = a_k$ for all $1 \le k \le n$ (the $zP_{n-1}$ and $zR_{n-1}$ terms are the only terms of degree $n$ in $z$). Hence, for each combination of $t_k$, no other upper Hessenberg Toeplitz matrix with $t_k \in Pt$ and subdiagonal $1$ has the same characteristic polynomial.
\end{proof}

% \begin{proposition}
%     The set of all characteristic polynomials of all matrices in $\mathcal{M}_n$ has cardinality $3^n$.
% \end{proposition}
% \begin{proof}
%     By Corollary~\ref{cor:onlysim}, each matrix $\M_n \in \mathcal{M}_n$ is similar to only one other matrix $\M_n^{+} \in \mathcal{M}_n$. Since $\M_n \ne \M_n^{+}$, and $\M_n$ shares a characteristic polynomial with $\M_n^{+}$; the cardinality of the set of characteristic polynomials must be half the cardinality of $\mathcal{M}_n$.
% \end{proof}

\begin{proposition}
    \label{prop:maxheight_UHT}
    The characteristic height of $\M_n \in \MMZ$ is maximal when $t_k = -1$ for $1 \le k \le n$.
\end{proposition}
\begin{proof}
    Following from Proposition~\ref{prop:maxheight}, the entries in the $i$th row and $i+k-1$-th column for $1 \le i \le n-k+1$ correspond to $t_k$, after substituting $s=1$ we find
    $t_k = -1$ gives the maximal characteristic height.
\end{proof}
\begin{remark}
We will see that other matrices also have maximal characteristic height.  Therefore the matrix of Proposition~\ref{prop:maxheight_UHT} is not the only such, but it is an interesting one.
\end{remark}

\begin{remark}
We do not use in any essential way that the population $P$ is just $\{-1, 0, 1\}$. The theorem is still true if $P$ is invariant with respect to all $s_k$ and contains elements of magnitude at most 1. And contains the entries $\pm 1$ so that the maximum height is in fact achieved.
\end{remark}

\begin{proposition}
    \label{prop:max_height_set}
    Let $F \subset \mathbb{R}$ be a closed and bounded set with $a = \min{F}$, $b = \max{F}$ and $\#F \ge 2$. Let $\M_n \in \MM{0}{n}{F}$. If $|a| \ge |b|$, $\M_n$ is of maximal characteristic height when $t_k = a$ for all $1 \le k \le n$. If $|b| \ge |a|$, $\M_n$ is of maximal characteristic height for $t_k = a$ for $k$ even, and $t_k = b$ for $k$ odd.
\end{proposition}
\begin{proof}
    First, consider the case when $|a| \ge |b|$. Since $a < b$ we find $a < 0$. Let $\overline{t}_k = -t_k$. Writing Proposition~\ref{prop:negativeheight} in terms of $\overline{t}_k$ gives
    \begin{subequations}
    \label{eq:polyrec_neg_t}
    \begin{align}
        p_{n,n} &= 1,\\
        p_{n,j} &= p_{n-1,j-1} + \sum_{k=1}^{n-j} \overline{t}_k p_{n-k,j} \quad\text{for}\quad 1 \le j \le n-1,\\
        p_{n,0} &= \sum_{k=1}^n \overline{t}_k p_{n-k,0}, \,\text{and}\\
        p_{0,0} &= 1\>.
    \end{align}
    \end{subequations}
    If all $\overline{t}_k$ are positive then $p_{n,j}$ must be positive for all $n$ and $j$. Hence, the maximal characteristic height is attained when $\overline{t}_k$ is maximal, or equivalently $t_k$ is minimal and negative. Thus $t_k = \min{F} = a$ gives maximal characteristic height.

    Next, consider when $|b| \ge |a|$. Since $a < b$ we find $b > 0$. By Proposition~\ref{prop:negativeheight} we know that the characteristic height of $\M_n$ is equal to the characteristic height of $-\M_n$. Rewriting Proposition~\ref{prop:charpolyrec2_UHT} for $-\M_n$ by substituting $p_{n,j}$ with $(-1)^{n-j} p_{n,j}$ we find the recurrence for the characteristic polynomial of $-\M_n$:
    \begin{subequations}
    \label{eq:polyrec_neg_mat}
    \begin{align}
        p_{n,n} &= 1,\\
        p_{n,j} &= p_{n-1,j-1} + \sum_{k=1}^{n-j} (-1)^{k-1} t_k p_{n-k,j} \quad\text{for}\quad 1 \le j \le n-1,\\
        p_{n,0} &= \sum_{k=1}^n (-1)^{k-1} t_k p_{n-k,0}, \,\text{and}\\
        p_{0,0} &= 1\>.
    \end{align}
    \end{subequations}
    Separating out the even and odd values of $k$ in the sums we can write the recurrence as
    \begin{subequations}
    \label{eq:polyrec2_odd_even}
    \begin{align}
        p_{n,n} &= 1,\\
        p_{n,j} &= p_{n-1,j-1} + \sum_{k \text{ odd}}^{n-j} t_k p_{n-k,j} - \sum_{k \text{ even}}^{n-j} t_k p_{n-k,j} \quad\text{for}\quad 1 \le j \le n-1,\\
        p_{n,0} &= \sum_{k \text{ odd}}^n t_k p_{n-k,0} - \sum_{k \text{ even}}^n t_k p_{n-k,0}, \,\text{and}\\
        p_{0,0} &= 1\>.
    \end{align}
    \end{subequations}
    The odd sums are maximal for $t_k = \max{F} = b$ and the even sums are maximal for $t_k = \min{F} = a$. Hence, the maximal characteristic height is attained for $t_k = b$ when $k$ is odd, and $t_k = a$ when $k$ is even.

    When $|a| = |b|$, equations~\eqref{eq:polyrec_neg_t} and~\eqref{eq:polyrec2_odd_even} are equivalent and the maximal height is attained both when $t_k = b$ for all $k$, and $t_k = b$ for $k$ odd and $t_k = a$ for $k$ even.
\end{proof}

\begin{proposition}
    \label{prop:maxheight_UHT2}
    $\M_n \in \MMZ$ also attains maximal characteristic height when $t_k = (-1)^{k-1}$ for $1 \le k \le n$.
\end{proposition}
\begin{proof}
    By Proposition~\ref{prop:max_height_set}, we have $F = \{-1, 0, {+1}\}$ with $a = -1$, and $b = +1$. Thus $\M_n$ is also of maximal characteristic height for $t_k = b = +1$ for odd values of $k$, and $t_k = a = -1$ for even values of $k$.
\end{proof}

% ============================================================================ %
% Maximal Characteristic Height Upper Hessenberg Toeplitz Matrices             %
% ============================================================================ %
\subsection{Matrices with maximal characteristic height}
\label{sec:MCHUHTM}
In this section we restrict our analysis to those matrices in \MMZ of maximal characteristic height. We denote this subset by $\MMM$. By proposition~\ref{prop:maxheight_UHT} this subset contains more than one element. Let $\tau_n$ be the characteristic height of \MMM (the height is the same for all matrices in \MMM) and for each $\M \in \MMM$ let $\mu_{n}(\M)$ be the largest degree of the term in the characteristic polynomial $\M$ which has 
\begin{align}
    \pm \tau_{n} = a_{\mu_{n}} = [\lambda^{\mu_{n}}]\left(\det\left(\lambda\mathbf{I} - \M\right)\right)
% \\
% &= a_0 + a_1\lambda + \cdots + a_{n-1}\lambda^{n-1} + \lambda^{n}
\end{align}
so
\begin{equation}
    \mathrm{charpoly}(\M) = a_0 + a_1\lambda + \cdots \pm\tau_n\lambda^{\mu_{n}} + a_{\mu_{n}+1}\lambda^{\mu_{n} + 1} + \cdots + \lambda^{n}
\end{equation}
and $|a_{\mu_{n}+1}| < \tau_n$. In Proposition~\ref{prop:unique_mu}, we prove that $\mu_{n}(\mathbf{M})$ is constant for $\MMM$. That is, the largest coefficient always appears at the same degree.

%let $\mu_n$ be the degree of the term of the characteristic polynomial of $\overline{\mathbf{M}}_n \in \MMM$ whose coefficient gives the height. In Proposition~\ref{prop:unique_mu} we will prove that $\mu_n$ the same for all matrices in \MMM. $\tau_n$ and $\mu_n$ and the number of matrices with maximal characteristic height for dimensions 2 to 10 are given in Table~\ref{tab:max_char_height}.

\begin{table}[h!]
  \begin{center}
    \begin{tabular}{cccc}
        \toprule
        $n$ & $\tau_n$ & $\mu_n$ & \# max char height\\ \midrule
         2 &     2 & 1 &  6\\
         3 &     5 & 1 &  6\\
         4 &    12 & 1 &  6\\
         5 &    27 & 1 &  6\\
         6 &    66 & 2 & 18\\
         7 &   168 & 2 & 18\\
         8 &   416 & 2 & 18\\
         9 & 1,008 & 2 & 18\\
        10 & 2,528 & 3 & 54 \\
        \bottomrule
    \end{tabular}
  \end{center}
  \label{tab:max_char_height}
  \caption{Maximum height, $\tau_n$, degree of term of characteristic polynomial corresponding to maximum height, $\mu_n$, and the number of matrices in \MMM for dimensions 2 to 10.}
\end{table}

% \begin{proposition}
%     \label{thm:max_height_bound}
%     The characteristic height, $\tau_n$ grows at least exponentially in $n$.
% \end{proposition}
% \begin{proof}
%     When $t_k = -1$ for $1 \le k \le n$, the characteristic height is maximal by Proposition~\ref{prop:maxheight_UHT}.
%     Equation~\eqref{eqn:them2_cor_UHT_c} from Proposition~\ref{prop:charpolyrec2_UHT} reduces to
%     \begin{equation}
%         p_{n,0} = \sum_{k=1}^n p_{n-k,0} = 2^{n-1}
%     \end{equation}
%     for $n \ge 1$ with $p_{0,0} = 1$ by equation~\eqref{eqn:them2_cor_UHT_d}. Thus, the maximal characteristic height must grow at least exponentially in $n$.
% \end{proof}

\begin{proposition}
    \label{prop:heightcoeffs}
    The characteristic height, $\tau_n$ is independent of $t_j$ for $j > n - \mu_n$.
\end{proposition}
\begin{proof}
    Let $\overline{P}_n$ be the characteristic polynomial of $\overline{\mathbf{M}}_n \in \MMM$.
    By Proposition~\ref{prop:funof}, $p_{n,\mu_n}$ is independent of $t_j$ for $j > n -  \mu_n$. Thus, $t_j$ for $j > n -  \mu_n$ only affects $p_{n,k}$ for $k <  \mu_n$. Since $\overline{\mathbf{M}}_n$ is of maximal height, $|p_{n,k}| \le |p_{n,  \mu_n}|$ for $k < \mu_n$ for all $t_j\in \{-1, 0, +1\}$ with $j > n -  \mu_n$.
\end{proof}

\begin{proposition}
    \label{prop:unique_mu}
    For fixed $n$, $\mu_n$ is the same for all $\overline{\mathbf{M}}_n \in \MMM$.
\end{proposition}
\begin{proof}
The characteristic polynomial of $\overline{\mathbf{M}}_n$ when $t_k = -1$ has the same coefficients as the characteristic polynomial of $\overline{\mathbf{M}}_n$ for $t_k = (-1)^{k-1}$ up to a sign change. By Proposition~\ref{prop:heightcoeffs}, changing any of the entries $t_j$ of $\overline{\mathbf{M}}_n$ for $j > n-\mu_n$ does not affect the value of $\mu_n$. Therefore $\mu_n$ is fixed.
\end{proof}

\begin{theorem}
\label{thm:count_max_height_uht}
\MMM contains $2\cdot3^{\mu_n}$ matrices.
\end{theorem}
\begin{proof}
    By Proposition~\ref{prop:maxheight_UHT}, $t_k = -1$ for $1 \le k \le n$ gives maximal characteristic height. By Proposition~\ref{prop:heightcoeffs}, any combination of $t_j \in \{-1, 0, {+1}\}$ for $j > n - \mu_n$ will not affect the characteristic height. Thus, the $3^{\mu_n}$ matrices with $t_k = -1$ for $1 \le k \le n-\mu_n$, and $t_k \in \{-1, 0, {+1}\}$ for $n - \mu_n +1 \le k \le n$ all have maximal characteristic height. Similarly, by Proposition~\ref{prop:maxheight_UHT2}, $t_k = (-1)^{k-1}$ for $1 \le k \le n$ gives maximal characteristic height. Again, by Proposition~\ref{prop:heightcoeffs}, $3^{\mu_n}$ matrices with $t_k = (-1)^{k-1}$ for $1 \le k \le n-\mu_n$, and $t_k \in \{-1, 0, {+1}\}$ for $n - \mu_n +1 \le k \le n$ all have maximal characteristic height.
\end{proof}

\subsection{More about characteristic polynomials of maximal height}
\label{sec:MHCP}
In this section we restrict our analysis to the matrix $\widetilde{\M}_n \in \MMM$ with $t_k = -1$ for all $k$. By Proposition~\ref{prop:maxheight_UHT}, $\widetilde{\M}_n$ is of maximal characteristic height.
Call these special characteristic polynomials $\widetilde{P}_n(z) = \det\left(z\mathbf{I}-\widetilde{\M}_n\right)$.
\begin{proposition}
    These characteristic polynomials $\widetilde{P}_n(z)$ satisfy the three-term recurrence relation
    \begin{equation}
        \widetilde{P}_{n+1}(z) = (z+2)\widetilde{P}_{n}(z) - z\widetilde{P}_{n-1}(z)
    \end{equation}
    with the initial conditions $\widetilde{P}_0(z) = 1$ and $\widetilde{P}_1(z) = 1+z$.
\end{proposition}
\begin{proof}
    By equation~\ref{eqn:thm1_UHT} and proposition~\ref{prop:maxheight_UHT}, and using the fact that $t_k=-1$, the polynomials satisfy the recurrence relation
    \begin{equation}
        \widetilde{P}_n(z) = z\widetilde{P}_{n-1}(z) + \sum_{k=1}^n \widetilde{P}_{n-k}(z)\>.
    \end{equation}
    Relabeling the indices in the sum, and using the same relationship for $\widetilde{P}_{n+1}(z)$, we have
    \begin{equation}
        \widetilde{P}_n(z) = z \widetilde{P}_{n-1}(z) + \sum_{k=0}^{n-1} \widetilde{P}_k(z)
    \end{equation}
    and
    \begin{align}
        \widetilde{P}_{n+1}(z) &= z \widetilde{P}_{n}(z) + \sum_{k=0}^{n} \widetilde{P}_k(z) \nonumber \\
        &= z\widetilde{P}_n(z) + \widetilde{P}_n(z) + \sum_{k=0}^{n-1} \widetilde{P}_k(z)\>.
    \end{align}
    Subtracting the previous equation gives the proposition.
\end{proof}
\begin{proposition}
    $\widetilde{P}_n(z)$ is of the form
    \begin{equation}
        \widetilde{P}_n(z) = z^n + p_{n,n-1}z^{n-1} + \cdots + p_{n,0}
    \end{equation}
    where each coefficient $p_{n,j}$ is positive for all $n$ and $j$.
\end{proposition}
\begin{proof}
    When $t_k = -1$ for $1 \le k \le n$, Proposition~\ref{prop:charpolyrec2_UHT} reduces to
    \begin{subequations}
    \label{eqn:max_height_char_poly_rec}
    \begin{align}
        p_{n,n} &= 1,\\
        p_{n,j} &= p_{n-1,j-1} + \sum_{k=1}^{n-j} p_{n-k,j}\quad\text{for}\quad 1 \le j \le n-1, \label{eqn:max_height_char_poly_rec_b}\\
        p_{n,0} &= \sum_{k=1}^n p_{n-k,0},
 \,\text{and} \label{eqn:max_height_char_poly_rec_c}\\
        p_{0,0} &= 1 \label{eqn:max_height_char_poly_rec_d}\>.
    \end{align}
    \end{subequations}
    Since $p_{0,0}$ is positive, and all coefficients in the above equations are positive, $p_{n,j}$ must be positive for all $n$ and $j$.
\end{proof}

\begin{proposition}
    \label{prop:genfun}
    The generating function of the sequence $(p_{i,i}, p_{i+1,i}, \ldots)$ for all $i \ge 0$ is
    \begin{equation}
        G_{i}(x) = \bigg (\frac{1-x}{1-2x} \bigg)^{i+1} \>.
    \end{equation}
According to links from~\cite{oeisA062110}, this is the generating function for so-called \textsl{weak compositions}: therefore, $p_{n,k}$ gives the number of compositions of $n$ with exactly $k$ zeros.  See also~\cite{janjic2018words}.
\end{proposition}
\begin{proof}
Omitted for length. See the arXiv version of this paper for details.
    \end{proof}

\begin{remark}
    The coefficients $p_{n,k}$ are given by the OEIS sequence \href{http://oeis.org/A105306}{A105306}~\cite{oeisA105306} and \href{http://oeis.org/A062110}{A062110}~\cite{oeisA062110} for the ``number of directed column-convex polynomials of area $n$, having the top of the right-most column at height $k$.'' We have $p_{n,k} = T_{n+1,k+1}$ where
    \begin{equation}
        T_{n,k} = 
        \begin{cases}
            \displaystyle\sum_{j=0}^{n-k-1} \binom{k+j}{k-1} \binom{n-k-1}{j} & \text{if } k < n\\
            \hfil 1 & \text{if } k = n
        \end{cases}
    \end{equation}
    Maple ``simplifies'' this to
    \begin{equation}
      T_{n,k} =
      \begin{cases}
          kF \!
          \left(
              \begin{array}{c|c}
                  k+1, k+1-n & \multirow{2}{*}{$-1$} \\
                  2 &
              \end{array}
          \right) & \text{if } n \ne k\\
          \hfil 1 & \text{if } n = k
        \end{cases}
    \end{equation}
    where $F(\cdot)$ is the hypergeometric function defined as
    \begin{equation}
        F \!
          \left(
              \begin{array}{c|c}
                  a, b & \multirow{2}{*}{$z$} \\
                  c &
              \end{array}
          \right)
          =
          \sum_{n=0}^{\infty} \frac{a^{\overline{n}} b^{\overline{n}}}{c^{\overline{n}}} \frac{z^n}{n!}
    \end{equation}
    where $q^{\overline{n}}$ ($q$ to the $n$ rising) is $q\cdot(q+1)\cdots(q + n - 1)$.
As stated previously, these also count the number of weak compositions of $n$ with exactly $k$ zeros.
\end{remark}
\begin{remark}
    This implies that an exact formula for the coefficient $p_{n,k}$ is
    \begin{equation}
        p_{n,k} = (k+1)\sum_{m=0}^{n-k-1} \frac{(k+2)^{\overline{m}}(n-k-1)^{\overline{m}}}{2^{\overline{m}} m!}\>.
    \end{equation}
    It is possible that the approximate formula for $\mu_n$ given earlier, placed into this formula, would allow detailed understanding of the asymptotics of $\tau_n$.
\end{remark}
\begin{proof}
The proof is again omitted for length.
\end{proof}
\begin{proposition}
    The characteristic polynomials have the ordinary generating function
    \begin{equation}
       \widetilde{G} =  {\frac {1-x}{z{x}^{2}- \left( z+2 \right) x+1}} = \sum_{n\ge0} \widetilde{P}_n(z) x^n\>.
    \end{equation}
\end{proposition}
\begin{proof}
    Straightforward from the recurrence relation above: $\widetilde{G} - (z+2)x\widetilde{G} + zx^2\widetilde{G} = p_0(z) + (\widetilde{P}_1(z)-(z+2)\widetilde{P}_0(z))x = 1-x$.
\end{proof}
\begin{theorem}
    \label{thm:max_height_bound}
    The maximum characteristic height, $\tau_n$, of any upper Hessenberg Bohemian with population $\{-1,0,1\}$ lies between the following bounds:
    \begin{equation}
        \frac{F_{2n+1}}{n+1} < \tau_n < F_{2n+1}\>.
    \end{equation}
    Here $F_k$ is the $k$th Fibonacci number, with the conventional numbering given by $F_{n+1}=F_n+F_{n-1}$ with $F_0 = 0$ and $F_1 = 1$.
\end{theorem}
To prove this, we first establish a lemma, using the recurrence relation.
\begin{lemma}
The value of $\widetilde{P}_n(z)$ when $z=1$ is a Fibonacci number, namely $\widetilde{P}_n(1) = F_{2n+1}$. 
\end{lemma}
\begin{proof} (of lemma):
This follows from the recurrence relation $\widetilde{P}_{n+1}(z) = (z+2)\widetilde{P}_n(z) - z \widetilde{P}_{n-1}(z)$ on substituting $z=1$ to get $\widetilde{P}_{n+1}(1) = 3\widetilde{P}_n(1) - \widetilde{P}_{n-1}(1)$ with $\widetilde{P}_0(1) = 1$ and $\widetilde{P}_1(1) = 2$. Standard methods for solving recurrence relations show that $\widetilde{P}_n(1) = \left( \sqrt {5}/10+1/2 \right)  \left( 3/2+\sqrt {5}/2
 \right) ^{n}+ \left( -\sqrt {5}/10+1/2 \right)  \left( 3/2-\sqrt {5}/2 \right) ^{n}
$ and this is seen by inspection to be $F_{2n+1}$, because $3/2+\sqrt{5}/2 = (1/2+\sqrt{5}/2)^2$.
\end{proof}
\begin{proof} (of Theorem):
    We have seen that the characteristic polynomial has positive coefficients.  Therefore $\tau_n < \widetilde{P}_n(1) = \sum_{k=0}^n p_{n,k}$. By the previous lemma, $\widetilde{P}_n(1) = F_{2n+1}$.  This establishes the upper bound. The lower bound follows since the arithmetic mean of the coefficients cannot be larger than the maximum, and indeed must be smaller since at least one coefficient of $\widetilde{P}_n(z)$ is $1$ because the polynomial is monic.
\end{proof}

\begin{conjecture}
    The maximum characteristic height, $\tau_n$, approaches the exponentially growing function $C F_{2n+1}/\sqrt{n+1}$ as $n \to \infty$ for some constant $C$. Our experiments indicate that $C \doteq 0.7701532$.
\end{conjecture}
\begin{remark}
    This limit is illustrated in Figure~\ref{fig:maxheightratiolog}, motivating this conjecture. This conjectured behaviour seems to be a constant times the geometric mean of the lower and upper bounds of Theorem~\ref{thm:max_height_bound}. Our experimental evidence suggests that the relative error is $O(1/(n+1))$, but the detailed behaviour is very interesting, and reminiscent of the images of $\sin(n)$ in~\cite{Hardin:1990:thousand}.
\end{remark}
\begin{figure}[h]
    \centering
%    \includegraphics[width=\textwidth]{Figures/MaxHeightLogDiff_50000.pdf}
%    \caption{The points are $\log{\tau_{n+1}} - \log{\tau_n}$ for $n$ from 0 to 50,000 where $\tau_n$ is the maximal characteristic height of \MMZ (i.e. when $t_k = -1$, for example). The solid line is $\log(1 + \varphi)$ where $\varphi$ is the golden ratio.}
    \includegraphics[width=\textwidth]{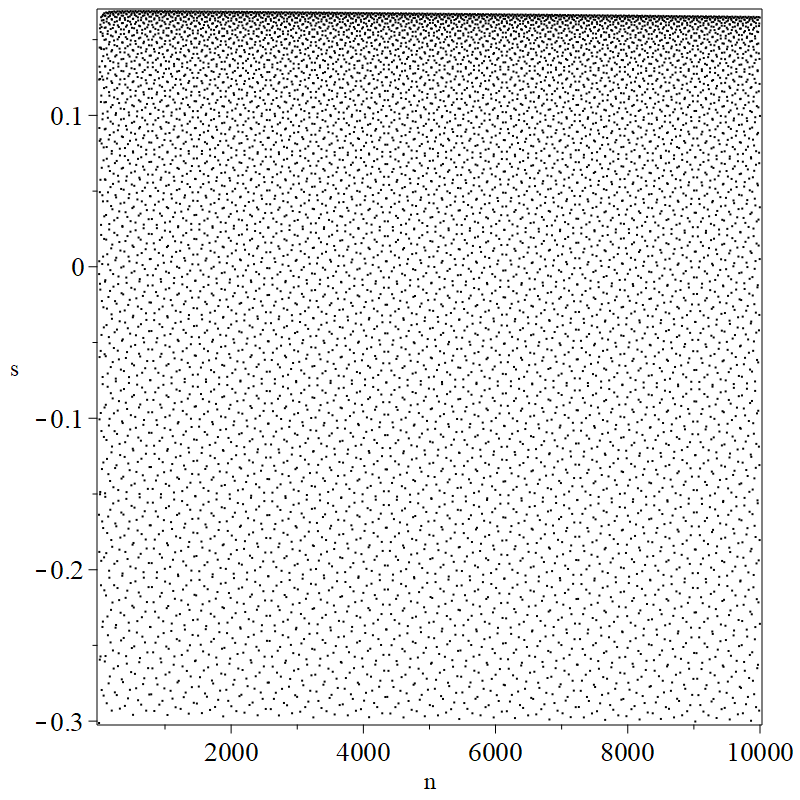}
    \caption{Asymptotic behaviour of the maximum characteristic height $\tau_n$.  With the conjectured $G_n= C\cdot F_{2n+1}/\sqrt{n+1}$ for $C \doteq 0.7701532$ and $F_k$ the $k$th Fibonacci number, we graph $s = (n+1)\cdot\left(G_n/\tau_n - 1\right) $ which is $(n+1)$ times the relative difference $G_n/\tau_n - 1$.  This shows evidence that $\tau_n = G_n\left(1+\tilde{O}((n+1)^{-1})\right)$, where we leave the meaning of the $\tilde{O}$-symbol carefully unspecified in this conjecture.}
    \label{fig:maxheightratiolog}
\end{figure}

% ------------------------- %
% A 2D recurrence           %
% ------------------------- %
\begin{proposition}
    The characteristic polynomial of $\widetilde{\M}_n$ is
    \begin{align}
    \widetilde{P}_n(z)
    &= \sum_{\ell = 0}^{\lfloor \sfrac{n}{2} \rfloor}
    {n \choose 2\ell}
    \bigg(\frac{z}{2} + 1\bigg)^{n - 2\ell}\bigg(1 + \frac{z^2}{4}\bigg)^{\ell} \nonumber\\
    &\qquad + \dfrac{z}{2}\sum_{\ell = 0}^{\lfloor \frac{n-1}{2}\rfloor}
    {n \choose 2\ell + 1}
    \bigg(\frac{z}{2} + 1\bigg)^{n - 2\ell - 1}\bigg(1 + \frac{z^2}{4}\bigg)^{\ell} \>.
    \end{align}
\end{proposition}
% This proposition can be proved in several ways. We choose below to think of $z \in \mathbb{C} \setminus \{\pm 2i\}$, for a reason that will become clear. Since the end result is a polynomial in $z$, proving the formula for $z\neq \pm 2i$ will recover the exceptional cases by continuity.

% Another equally valid approach would be to think of $z$ as being transcendental and noting that the characteristic polynomial of $\widetilde{\M}_{n}$ has integer coefficients.
\begin{proof}
Straightforward but tedious analysis of the linear recurrence relation or, equivalently, the generating function.
\end{proof}

% ============================================================================ %
% A Connection with Compositions                                               %
% ============================================================================ %
\subsection{A Connection with Compositions}
Consider the case with symbolic entries $t_i$, and subdiagonals $-1$ for convenience with minus signs in the formulae.  For instance, the $5$ by $5$ example upper Hessenberg Toeplitz matrix is
\begin{equation}
    \textbf{M}_5 =  \left[ \begin {array}{ccccc} t_{{1}}&t_{{2}}&t_{{3}}&t_{{4}}&t_{{5}}
\\ \noalign{\medskip}-1&t_{{1}}&t_{{2}}&t_{{3}}&t_{{4}}
\\ \noalign{\medskip}0&-1&t_{{1}}&t_{{2}}&t_{{3}}\\ \noalign{\medskip}0
&0&-1&t_{{1}}&t_{{2}}\\ \noalign{\medskip}0&0&0&-1&t_{{1}}\end {array}
 \right] 
\>.
\end{equation}
In this section we consider what happens when we take determinants $P_n(z) = \det(z\mathbf{I} - \textbf{M}_n)$.
Examining $P_0(0)$, $P_1(0)$, $P_2(0)$, $P_3(0)$, and $P_4(0)$, and in particular $P_k(0)$ (i.e. $\det (-\mathbf{M}_k)$) we see that
\begin{align}
    P_0(0) &= 1 \mathrm{\ by\  convention}\\
    P_1(0) &= t_1\\
    P_2(0) &= t_1^2 + t_2\\
    P_3(0) &= t_1^3 + 2t_1t_2 + t_3\\
    P_4(0) &= t_1^4 + 3t_1^2t_2 + 2t_1t_3 + t_2^2 + t_4\>.
\end{align}
One may interpret these (looking at the subscripts) as \textsl{compositions}: $2 = 1 + 1 = 2$; $3 = 1 + 1 + 1 = 1 + 2 = 2 + 1 = 3$; $4 = 1 + 1 + 1 + 1 = 2 + 1 + 1 = 1 + 2 + 1 = 1 + 1 + 2 = 1 + 3 = 3 + 1 = 2 + 2 = 4$. The number of compositions of $n$ is $2^{n-1}$, which we get if all $t_j = 1$.
The paper~\cite{shattuck2013combinatorial} shows a connection between compositions, Hessenberg matrices, and Fibonacci numbers. In this section we merely remark on the connection.

From the Wikipedia entry on composition (combinatorics), ``a composition of an integer $n$ is a way of writing $n$ as the sum of a sequence of strictly positive integers"~\cite{wiki_composition}. We mentioned earlier that $p_{n,k}$ is reported in links from~\cite{oeisA062110} to be the number of ``weak compositions'' of $n$ with exactly $k$ zeros, although we have not proved that here.  The notion of a ``weak'' composition extends the notion of composition.

One may interpret the recurrence relation
\begin{equation}
    p_{n, 0} = \sum_{k = 1}^{n}t_kp_{n-k, 0}
\end{equation}
from Proposition~\ref{prop:charpolyrec2_UHT} as saying that to generate a composition of $n$, you get the composition of $n-k$ and then add the number ``$k$" to them; adding these together gives all compositions. For example, when $n = 5$ we have $p_{0, 0} = 1$, $p_{1, 0} = t_1$, $p_{2, 0} = t_1^2 + t_2$, $p_{3, 0} = t_1^3 + 2t_1t_2 + t_3$, and $p_{4, 0} = t_1^{4} + 3t_1^2t_2 + 2t_1t_3 + t_2^2 + t_4$. Then
\begin{align*}
    p_{5, 0} &= t_1 p_{4, 0} + t_2p_{3, 0} + t_3 p_{2, 0} + t_4p_{1, 0} + t_5p_{0, 0} \nonumber\\
    &= t_1^5 + 3t_1^3t_2 + 2t_1^2t_3 + t_1t_2^2 + t_1t_4 + t_2t_1^3 + 2t_1t_2^2 + t_2t_3 + t_1^2t_3 + t_2t_3 + t_4t_1 + t_5 \nonumber \\
    &= t_1^5 + 4t_1^3t_2 + 3t_1^2t_3 + 3t_1t_2^2 + 2t_1t_4 + 2t_2t_3 + t_5 \>.
\end{align*}

\begin{remark}
This determinant also contains the whole characteristic polynomial. Simply replace $t$, with $t_1 - z$ and we get $\det \left(\mathbf{M}_n - z\mathbf{I}\right) = (-1)^{n}P_n$. This suggests that ``compositions with all parts bigger than 1" can be used to generate all compositions. This fact is well-known.
{The combinatorial analysis of this recurrence formula is not quite trivial.}
\end{remark}

\section{Motivating interest in Bohemians\label{sec:Motivation}}
In this section we discuss some details of our motivations for investigating these matrices.
Typical computational puzzles arise for us on asking simple-looking questions such as ``how many $6 \times 6$ matrices with the population $\{-1, 0, {+1}\}$ are singular.'' Such a question helps to understand the probability of encountering singularity when matrices are drawn ``at random'' from such a collection.  The answer is not known as we write this, although we can give a probabilistic estimate ($0.205$ after $20,000,000$ sample determinants\footnote{4103732 singular matrices out of twenty million sampled.}): brute computation seems futile to find the exact number, because there are $3^{36} \doteq 1.7\times10^{17}$ such matrices. We do know the answers up to size five by five: The number of $n$ by $n$ singular matrices with population $\{-1, 0, {+1}\}$ is, for $n=1$, $2$, $3$, $4$, and $5$,
just $1$, $33$, $7,875$, $15,099,201$, and $237,634,987,683$.
This represents fractions of their numbers ($3^{n^2}$) of $0.333$, $0.407$, $0.400$, $0.351$, and $0.280$,
respectively.

Even though we do not yet know the exact answers to these questions, such matrix families can be both useful and interesting. For instance, one may use discrete optimization over a family to look for improved growth factor bounds~\cite{higham2018bohemian}. Matrices with the population $\{-1, 0, {+1}\}$ have minimal height over all integer matrices; finding a matrix in this family which has a given polynomial $p(\lambda) \in \mathbb{Z}[\lambda]$ as characteristic polynomial identifies a so-called ``minimal height companion matrix'', which may confer numerical benefits.

Recently the study of eigenvalues of structured Bohemians (e.g.~tridiagonal, complex symmetric) has been undertaken and several puzzling features are seen resulting from extensive experimental computations. For instance, some of the images at \href{http://www.bohemianmatrices.com/gallery}{bohemianmatrices.com/gallery} show common features including ``holes''.

Visible features of graphs of roots and eigenvalues from structured families of polynomials and matrices have been previously studied. One well-known set of polynomials whose roots produce interesting pictures are the Littlewood polynomials,
\begin{equation}
    p(x) = \sum_{i = 0}^{n}a_{i}x^{i} \>,
\end{equation}
where $a_{i} \in \{-1, {+1}\}$. These polynomials have been studied in \cite{baez2009beauty}, \cite{borwein2012computational}, \cite{borwein2001visible}, and \cite{borwein1997polynomials}. Similarly, polynomials with coefficients $\{0,1\}$ (also called Newman polynomials) have been studied by Odlyzko and Poonen~\cite{odlyzko1993zeros}.  

The eigenvalues of bounded height marices raise many questions, ranging from whether the sets are (ultimately, as $n \to \infty$) \textsl{fractals} and what the boundaries of the sets are, to questions about the holes in the images and their possible connection to various properties. Answers to some of these questions, particularly the ones involving the holes, have been shown to have some significance in number theory~\cite{beaucoup1998multiple}. Roots of other polynomials have also been visualized; for more, see Christensen's\footnote{\url{https://jdc.math.uwo.ca/roots/}} and J{\"o}rgenson's\footnote{\url{http://www.cecm.sfu.ca/~loki/Projects/Roots/}} web pages.

Corless used a generalization of the Littlewood polynomial (to Lagrange bases). In his paper \cite{corless2004generalized}, he gave a new kind of companion matrix for polynomials expressed in a Lagrange basis. He used generalized Littlewood polynomials as test problems for his algorithm.

``The Bohemian Eigenvalue Project" was first presented as a
poster~\cite{eccadposter2015} at the East Coast Computer Algebra Day (ECCAD) 2015. The poster focused on preliminary results and many of the questions raised when visualizing the distributions of Bohemian eigenvalues over the complex plane. In particular, the poster focused on ``eigenvalue exclusion zones'' (i.e. distinct regions within the domain of the eigenvalues where no eigenvalues exist), computational methods for visualizing eigenvalues, and some results on eigenvalue conditioning over distributions of random matrices.

In Chan's Master's thesis~\cite{chan2016comparison}, she extended Piers W.~Lawrence's construction of the companion matrix for the Mandelbrot polynomials~\cite{corless2013largest, corlessMandelbrot} to other families of polynomials, mainly the Fibonacci-Mandelbrot polynomials and the Narayana-Mandelbrot polynomials. What is relevant here about this construction is that these matrices are upper Hessenberg and contain entries from a constrained set of numbers: $\{-1, 0\}$, and therefore fall under the category of being Bohemian upper Hessenberg. Both the Fibonacci-Mandelbrot matrices and Narayana-Mandelbrot matrices are also Bohemian upper Hessenberg, but the set that the entries draw from is $\{-1, 0, {+1}\}$. At the time of submission for Chan's Master's thesis, the largest number of eigenvalues successfully computed (using a machine with 32 GB of memory) were $32,767$, $17,710$, and $18,559$ for the Mandelbrot, Fibonacci-Mandelbrot, and Narayana-Mandelbrot matrices, respectively. This makes the \nth{16} Mandelbrot matrix the ``largest" Bohemian that we have solved at the time we write this paper.

% \begin{figure}
%     \centering
%     \includegraphics[width=6in]{Figures/mand_16.png}
%     \caption{All $32,767$ eigenvalues to $\mathbf{M}_{16}$.}
% \label{fig:mand_16}
% \end{figure}

These new constructions led Chan and Corless to a new kind of companion matrix for polynomials of the form $c(z) = z a(z) b(z) + c_0$.  A first step towards this was first proved using the Schur complement in \cite{chan2017new}.
%This paper caught the attention of Donald E.~Knuth; he did not like this proof and gave Chan and Corless a better proof which takes advantage of the fact that the determinant is linear in the first row, in \cite{chan2019algebraic}.
Knuth then suggested that Chan and Corless look at the Euclid polynomials~\cite{chan2017minimal}, based on the Euclid numbers.
It was the success of this construction that led to the realization that this construction is general, and gives a genuinely new kind of companion matrix.
%This furthers the original generalization of the new kind of companion matrix from multiplying two functions ($a(z)b(z)$ of the previous form) to an unlimited number of functions.
Similar to the previous three families of matrices, the Euclid matrices are also upper Hessenberg and Bohemian, as the entries are comprised from the set $\{-1, 0, +1\}$. In addition, an interesting property of these companion matrices is that their inverses are also Bohemian with the same population, a property which we call ``the matrix family having \emph{rhapsody}~\cite{chan2019algebraic}.''

As an extension of this generalization, Chan et al.~\cite{chan2019algebraic} showed how to construct linearizations of matrix polynomials, particularly of the form $z\mathbf{a}(z)\mathbf{d}_0 + \mathbf{c}_0$, $\mathbf{a}(z)\mathbf{b}(z)$, $\mathbf{a}(z) +\mathbf{b}(z)$ (when $\mathrm{deg}(\mathbf{b}(z)) < \mathrm{deg}(\mathbf{a}(z))$, and $z\mathbf{a}(z)\mathbf{d}_0\mathbf{b}(z) + \mathbf{c}_0$, using a similar construction.

% ============================================================================ %
% Conclusion                                                                   %
% ============================================================================ %
\section{Concluding Remarks}
The class of upper Hessenberg Bohemians gives a useful way to study Bohemians in general. This is an instance of Polya's adage ``find a useful specialization"~\cite[p.~190]{polya2014solve}. Because these classes are simpler than the general case, we were able to establish several theorems. Note that the three families $\HH{0}{n}{\{0, {+1}\}}$, $\HH{0}{n}{\{-1, {+1}\}}$, and $\ZH{0}{n}{\{-1, 0, {+1}\}}$ are all subfamilies of $\HH{0}{n}{\{-1, 0, {+1}\}}$.

% In this paper we have introduced two new formulae for computing the characteristic polynomials of upper Hessenberg matrices. Our first formula, given in Theorem~\ref{thm:charPolyRec1}, also computes the characteristic polynomials recursively. Our second formula, given in Theorem~\ref{thm:charPolyRec2}, computes the coefficients recursively.

We extended the formulae for the characteristic polynomials to upper Hessenberg Toeplitz matrices in Proposition~\ref{prop:maxheight_UHT} and Proposition~\ref{prop:maxheight_UHT2}. In Proposition~\ref{thm:max_height_bound} we showed that the maximal characteristic height of upper Hessenberg matrices in $\HH{0}{n}{\{-1, 0, {+1}\}}$ is at least $F_{2n+1}/(n+1)$. In Theorem~\ref{thm:count_max_height_uht} we show that the number of upper Hessenberg Toeplitz matrices of maximal height in \MM{0}{n}{\{-1, 0, {+1}\}} is $3\cdot 2^{\mu_n}$ where $\mu_n$ is the maximum degree of the coefficient of the characteristic polynomial whose coefficient, in absolute value, is the height.
We noted several connections to combinatorial works, such as~\cite{janjic2018words}.

We also explored some properties of zero diagonal upper Hessenberg Bohemians. In Theorem~\ref{thm:zero_diag_UH}, we show that the subset of these matrices that are normal are always symmetric, $w_{j}$-skew symmetric for some fixed $1 \leq j \leq m$, or $w_{j}$-skew circulant. In Theorem \ref{thm:not_stable}, we showed that no $\H \in \ZH{\theta_{k}}{n}{P}$ is stable.

% Many puzzles remain. Perhaps the most striking is the angular appearance of the set $\mathbf{\Lambda}(\HH{0}{n}{P})$ of eigenvalues of $\HH{0}{n}{P}$, \textcolor{red}{such as in Figures~\ref{fig:UH_6} and \ref{fig:UH_6_0_Diag}}. General matrices have eigenvalues asymptotic to a (scaled) disc~\cite{tao2017random}; our computations suggest that as $n \to \infty$, $\sfrac{\mathbf{\Lambda}(\HH{0}{n}{P})}{n^{\sfrac{1}{2}}}$ tends to an irregular hexagonal shape, rather than a disk. More, the density does not appear to be approaching uniformity. Further, the boundary is irregular, with shapes suggestive of what is popularly known as the ``dragon curve" (in reverse---these delineate where the eigenvalues are absent, near the edge). We have no explanation for this.

Searching for nilpotent matrices in various classes of Bohemians turns up several puzzles.
% We give some preliminary results here in Table~\ref{tab:nilpotents}, but leave this mostly to future work.  
For instance, it seems clear from our experiments that the only nilpotent matrix in $\HH{0}{n}{\{0, {+1}\}}$ is the (transpose of the) complete Jordan block of $n$ zero eigenvalues; contrariwise the irregular behaviour for $\HH{0}{n}{\{-1, {+1}\}}$ is very puzzling.

% \begin{table}[ht]
%     \centering
%     \begin{tabular}{cccc}
%           \toprule
%      $n$ & $\ZH{0}{n}{\{-1, 0, {+1}\}}$ & $\HH{0}{n}{\{0, {+1}\}}$ & $\HH{0}{n}{\{-1, {+1}\}}$  \\
%         \midrule
%     2   & 1 & 1 & 2 \\
%     3   & 3& 1 & 0 \\
%     4   & 21 & 1 & 0 \\
%     5   & 271 & 1 & 0 \\
%     6   & 9,075 & 1 & 324 \\
%           \bottomrule
%     \end{tabular}
%     \caption{The numbers of nilpotent matrices for various populations and dimensions}
%     \label{tab:nilpotents}
% \end{table}

% Data for this paragraph from the
% file nilpotents.mw in the subdirectory
% Dropbox/Polys_for_Rob
%
% Considering \textsl{general} Bohemian matrices with population $\{-1, 0, {+1}\}$, so that there are $3^{n^2}$ such matrices, we find that there are $1$, $9$, $481$, $148,817$, and $243,782,721$ nilpotent matrices at dimensions $1$ through $5$ inclusive.  We can fit this experimentally with the formula $\exp(0.5 + 0.38n + 0.23n^2)$, or something like $1.26^{n^2}$, which vanishes very quickly compared to $3^{n^2}$. This formula predicts that for $n=6$ the probability of finding a nilpotent matrix is about $2.75\times 10^{-14}$.  It would be gratifying to have a better understanding of the number of nilpotent matrices in a family.

\section*{Acknowledgements}
The calculations and images presented here were in part made possible using AMD Threadripper workstations provided by the Department of Applied Mathematics at Western University.
We acknowledge the support of the Ontario Graduate Institution, The National Science \& Engineering Research Council of Canada, the University of
Alcal\'a, the Rotman Institute of Philosophy, the Ontario Research Centre of
Computer Algebra, and Western University. Part of this work was developed
while R.~M.~Corless was visiting the University of Alcal\'a, in the frame of the
project Giner de los Rios. L.~Gonzalez-Vega, J.~R.~Sendra and J.~Sendra are
partially supported by the Spanish Ministerio de Econom\'\i a y Competitividad
under the Project MTM2017-88796-P.

\bibliographystyle{siamplain}  % siam
\bibliography{bibliography}

\begin{thebibliography}{10}

\bibitem{baez2009beauty}
{\sc J.~Baez}, {\em The beauty of roots}, Available at:
  https://johncarlosbaez.wordpress.com/2011/12/11/the-beauty-of-roots/,
  (2011).

\bibitem{beaucoup1998multiple}
{\sc F.~Beaucoup, P.~Borwein, D.~W. Boyd, and C.~Pinner}, {\em Multiple roots
  of $[- 1, 1]$ power series}, Journal of the London Mathematical Society, 57
  (1998), pp.~135--147.

\bibitem{borwein2012computational}
{\sc P.~Borwein}, {\em Computational excursions in analysis and number theory},
  Springer Science \& Business Media, 2012.

\bibitem{borwein2001visible}
{\sc P.~Borwein and L.~J{\"o}rgenson}, {\em Visible structures in number
  theory}, The American Mathematical Monthly, 108 (2001), pp.~897--910.

\bibitem{borwein1997polynomials}
{\sc P.~Borwein and C.~Pinner}, {\em Polynomials with $\{$0,+1,-1$\}$
  coefficients and a root close to a given point}, Canadian Journal of
  Mathematics, 49 (1997), pp.~887--915.

\bibitem{briat2017sign}
{\sc C.~Briat}, {\em Sign properties of {Metzler} matrices with applications},
  Linear Algebra and its Applications, 515 (2017), pp.~53--86.

\bibitem{cahill2002fibonacci}
{\sc N.~D. Cahill, J.~R. D'Errico, D.~A. Narayan, and J.~Y. Narayan}, {\em
  Fibonacci determinants}, The College Mathematics Journal, 33 (2002),
  pp.~221--225.

\bibitem{chan2016comparison}
{\sc E.~Y.~S. Chan}, {\em A comparison of solution methods for
  {M}andelbrot-like polynomials}, Electronic Thesis and Dissertation
  Repository,  (2016).
\newblock \url{https://ir.lib.uwo.ca/etd/4028}.

\bibitem{chan2017new}
{\sc E.~Y.~S. Chan and R.~M. Corless}, {\em A new kind of companion matrix},
  Electronic Journal of Linear Algebra, 32 (2017), pp.~335--342.

\bibitem{chan2017minimal}
{\sc E.~Y.~S. Chan and R.~M. Corless}, {\em Minimal height companion matrices
  for {E}uclid polynomials}, Mathematics in Computer Science,  (2018),
  \url{https://doi.org/10.1007/s11786-018-0364-2}.

\bibitem{chan2019algebraic}
{\sc E.~Y.~S. Chan, R.~M. Corless, L.~Gonzalez-Vega, J.~R. Sendra, and
  J.~Sendra}, {\em Algebraic linearizations of matrix polynomials}, Linear
  Algebra and its Applications, 563 (2019), pp.~373--399.

\bibitem{corless2004generalized}
{\sc R.~M. Corless}, {\em Generalized companion matrices in the {L}agrange
  basis}, in Proceedings EACA, Santander, Spain: Universidad de Cantabria,
  2004, pp.~317--322.

\bibitem{corlessMandelbrot}
{\sc R.~M. Corless and P.~W. Lawrence}, {\em Mandelbrot polynomials and
  matrices}.
\newblock In preparation.

\bibitem{corless2013largest}
{\sc R.~M. Corless and P.~W. Lawrence}, {\em The largest roots of the
  {M}andelbrot polynomials}, in Computational and Analytical Mathematics,
  Springer, 2013, pp.~305--324.

\bibitem{eccadposter2015}
{\sc R.~M. Corless and S.~Thornton}, {\em Visualizing eigenvalues of random
  matrices}, ACM Communications in Computer Algebra, 50 (2016), pp.~35--39,
  \url{https://doi.org/10.1145/2930964.2930969}.

\bibitem{corless2017bohemian}
{\sc R.~M. Corless and S.~E. Thornton}, {\em The {B}ohemian eigenvalue
  project}, ACM Communications in Computer Algebra, 50 (2016), pp.~158--160.

\bibitem{DeAlba:HLA:2013}
{\sc L.~M. DeAlba}, {\em Determinants and eigenvalues}, in Handbook of Linear
  Algebra, L.~Hogben, ed., Chapman and Hall/CRC, 2013, ch.~4.

\bibitem{Embree:HLA:2013}
{\sc M.~Embree}, {\em Pseudospectra}, in Handbook of Linear Algebra, L.~Hogben,
  ed., Chapman and Hall/CRC, 2013, ch.~23.

\bibitem{Fasi:2019:determinants}
{\sc M.~Fasi and G.~M.~N. Porzio}, {\em Determinants of normalized {B}ohemian
  upper {H}essenberg matrices}.
\newblock May 2019, \url{http://eprints.maths.manchester.ac.uk/2709/}.

\bibitem{gear1969simple}
{\sc C.~Gear}, {\em A simple set of test matrices for eigenvalue programs},
  Mathematics of Computation, 23 (1969), pp.~119--125.

\bibitem{Hall:HLA:2013}
{\sc F.~J. Hall and Z.~Li}, {\em Sign pattern matrices}, in Handbook of Linear
  Algebra, L.~Hogben, ed., Chapman and Hall/CRC, 2013, ch.~42.

\bibitem{Hardin:1990:thousand}
{\sc D.~Hardin and G.~Strang}, {\em A thousand points of light}, The College
  Mathematics Journal, 21 (1990), pp.~406--409,
  \url{https://doi.org/10.1080/07468342.1990.11973345},
  \url{https://doi.org/10.1080/07468342.1990.11973345},
  \url{https://arxiv.org/abs/https://doi.org/10.1080/07468342.1990.11973345}.

\bibitem{higham2018bohemian}
{\sc N.~Higham}, {\em Bohemian matrices in numerical linear algebra}.
\newblock Available at
  \url{http://www.maths.manchester.ac.uk/~higham/conferences/bohemian/higham_bohemian18.pdf}(June
  20, 2018).

\bibitem{janjic2018words}
{\sc M.~Janjic}, {\em Words and linear recurrences}, Journal of Integer
  Sequences, 21 (2018), p.~3.

\bibitem{jeffries1977matrix}
{\sc C.~Jeffries, V.~Klee, and P.~Van~den Driessche}, {\em When is a matrix
  sign stable?}, Canadian Journal of Mathematics, 29 (1977), pp.~315--326.

\bibitem{kilic2010generalized}
{\sc E.~Kilic and D.~Tasci}, {\em On the generalized {F}ibonacci and {P}ell
  sequences by {H}essenberg matrices}, Ars Combin, 94 (2010), pp.~161--174.

\bibitem{lettington2011fleck}
{\sc M.~C. Lettington}, {\em Fleck's congruence, associated magic squares and a
  zeta identity}, Functiones et Approximatio Commentarii Mathematici, 45
  (2011), pp.~165--205.

\bibitem{odlyzko1993zeros}
{\sc A.~M. Odlyzko}, {\em Zeros of polynomials with 0,1 coefficients}, in
  Algorithms Seminar, B.~Salvy, ed., no.~2130, December 1993, pp.~169--172,
  \url{http://citeseerx.ist.psu.edu/viewdoc/download?doi=10.1.1.47.9327&rep=rep1&type=pdf#page=175}.

\bibitem{polya2014solve}
{\sc G.~Polya}, {\em How to solve it: A new aspect of mathematical method},
  Princeton University Press, 2014.

\bibitem{shattuck2013combinatorial}
{\sc M.~Shattuck}, {\em Combinatorial proofs of determinant formulas for the
  {F}ibonacci and {L}ucas polynomials}, The Fibonacci Quarterly, 51 (2013),
  pp.~63--71.

\bibitem{oeisA062110}
{\sc N.~J.~A. Sloane}, {\em The on-line encyclopedia of integer sequences
  (a062110)}.
\newblock Published electronically at \url{https://oeis.org/A062110} (June 24,
  2018).

\bibitem{oeisA105306}
{\sc N.~J.~A. Sloane}, {\em The on-line encyclopedia of integer sequences
  (a105306)}.
\newblock Published electronically at \url{https://oeis.org/A105306} (June 20,
  2018).

\bibitem{tao2017random}
{\sc T.~Tao and V.~Vu}, {\em Random matrices have simple spectrum},
  Combinatorica, 37 (2017), pp.~539--553.

\bibitem{taussky1960matrices}
{\sc O.~Taussky}, {\em Matrices of rational integers}, Bulletin of the American
  Mathematical Society, 66 (1960), pp.~327--345.

\bibitem{taussky1961some}
{\sc O.~Taussky}, {\em Some computational problems involving integral
  matrices}, JOURNAL OF RESEARCH OF THE NATIONAL BUREAU OF STANDARDS SECTION
  B-MATHEMATICAL SCIENCES, 65 (1961), pp.~15--17.

\bibitem{CPDB}
{\sc S.~E. Thornton}, {\em The characteristic polynomial database}.
\newblock Available at \url{http://bohemianmatrices.com/cpdb} (Sept. 7, 2018).

\bibitem{wiki_composition}
{\sc Wikipedia}, {\em Composition (combinatorics)}.
\newblock Available at
  \url{https://en.wikipedia.org/wiki/Composition_(combinatorics)} (May 15,
  2019).

\end{thebibliography}

\end{document}